\documentclass[11pt,fleqn]{article} %[fleqn] move the equations to left

\usepackage[margin=1.0in]{geometry}
\usepackage{amsmath,amsthm,amssymb}
\usepackage{mathtools}
\usepackage{graphicx}
\usepackage{xcolor}
\usepackage{setspace}
\usepackage{multirow}
\usepackage{comment}
\usepackage{dsfont}
\usepackage{url}
\usepackage{tabularx}
\usepackage{amsfonts} % For math fonts, symbols and environm
\usepackage{algorithm}% http://ctan.org/pkg/algorithm
\usepackage{algpseudocode}% http://ctan.org/pkg/algorithmicx
\usepackage{hyperref}
\usepackage{appendix}
\usepackage{authblk}
\newtheorem{Theorem}{Theorem}

\newtheorem{Lemma}[Theorem]{Lemma}

\newtheorem{Experiment}{Experiment}

\def\T{{\mathrm{\scriptscriptstyle T}}}

\def\spacingset#1{\renewcommand{\baselinestretch}{#1}\small\normalsize}

\spacingset{1.25}
  \title{\vspace{-24pt}\textbf{Estimating dynamic transmission rates with a Black-Karasinski process in stochastic SIHR models using particle MCMC}}

\author[1,2]{Avery Drennan}
\author[2]{Jeffrey Covington}
\author[3]{Dan Han}
\author[1,2]{Andrew Attilio}
\author[1]{Jaechoul Lee}
\author[4]{Richard Posner}
\author[2]{Eck Doerry}
\author[2]{Joseph Mihaljevic}
\author[1]{Ye Chen\thanks{Correspondence to: Ye Chen. Email: ye.chen@nau.edu}}

\affil[1]{\small Department of Mathematics and Statistics, Northern Arizona University, Flagstaff, AZ, U.S.A.}
\affil[2]{\small School of Informatics, Computing, and Cyber Systems, Northern Arizona University, Flagstaff, AZ, U.S.A.}
\affil[3]{\small Department of Mathematics, University of Louisville, Louisville, KY, U.S.A.}
\affil[4]{\small Department of Biological Sciences, Northern Arizona University, Flagstaff, AZ, U.S.A.}

  \date{\today}
\begin{document}

\maketitle

\vspace{-12pt}
\begin{abstract}
Compartmental models are effective in modeling the spread of infectious pathogens, but have remaining weaknesses in fitting to real datasets exhibiting stochastic effects. We propose a stochastic SIHR model with a dynamic transmission rate, where the rate is modeled by the Black-Karasinski (BK) process — a mean-reverting stochastic process with a stable equilibrium distribution, making it well-suited for modeling long-term epidemic dynamics. To generate sample paths of the BK process and estimate static parameters of the system, we employ particle Markov Chain Monte Carlo (pMCMC) methods due to their effectiveness in handling complex state-space models and jointly estimating parameters. We designed experiments on synthetic data to assess estimation accuracy and its impact on inferred transmission rates; all BK-process parameters were estimated accurately except the mean-reverting rate. We also assess the sensitivity of pMCMC to misspecification of the mean-reversion rate. Our results show that estimation accuracy remains stable across different mean-reversion rates, though smaller values increase error variance and complicate inference results. Finally, we apply our model to Arizona flu hospitalization data, finding that parameter estimates are consistent with published survey data.
\end{abstract}

\vspace{12pt}

\section{Introduction}

% Why prove existence and uniqueness?
% What is new?
% Why use an OU process?
% Why so experiments with synthetic data?

Compartmental models, originally proposed by Kermack and McKendrick \cite{Kermack1927,Kermack1932}, are among the most widely used modeling approaches in epidemiology. 
They are valued for their efficiency, interpretability, and extensibility. However, performing inference and fitting these models to real data remain challenging. Introducing stochasticity helps overcome these difficulties and enhances interpretability when analyzing long-term disease dynamics, especially for the transmission rate  -- the key parameter that governs how quickly susceptible individuals become infected. 
Numerous factors modulate the transmission rate, and most vary over time \cite{London1973, Shaman2009}. Seasonality (temperature and humidity), changing contact patterns (school terms, holidays, mobility restrictions), shifts in population immunity (waning immunity, vaccination campaigns), pathogen evolution (immune-escape variants), and demographic or socioeconomic changes all act on different time-scales. Consequently, the transmission rate is not only time-varying, but has complex and unknown dynamics. 
An early approach to incorporating stochasticity in modeling the transmission rate $\beta_t$ was presented by Kalivianakis et al. \cite{Kalivianakis1994}. In their work, $\beta_t$ was modeled as a state variable within a state-space framework, and a Bayesian Filter was applied to infer the path of $\beta_t$. 
Bayesian filters, such as Particle Filters \cite{li2024real}, Ensemble Kalman Filters \cite{lal2021application, sun2023analysis}, Extended Kalman Filters \cite{song2021maximum}, and Iterative Filtering \cite{king2016statistical, fox2022real}, can effectively integrate stochastic compartmental models with data by inferring dynamic parameters like the stochastic transmission rate \cite{sheinson2014comparison, yang2014comparison}. However, the time-varying transmission rate is just one of several parameters that may need to be estimated, and many real-world applications require inferring dynamic and static parameters simultaneously.

Standard Bayesian filters estimate dynamic states with known static parameters. When static parameters are unknown, particle Markov Chain Monte Carlo (pMCMC) is ideal for jointly estimating both static parameters and dynamic states  \cite{andrieu2010particle}. Dureau et al. \cite{dureau2013BMpMCMC} introduced a stochastic differential equation (SDE) notation to model $\beta_t$ as a stochastic process and employed pMCMC techniques to infer the model parameter which includes the static parameters and sample path of $\beta_t$ jointly.  Similar applications are found in \cite{endo2019,camacho2015temporal,funk2018real,Cazelles2018,spannaus2022inferring,cazelles2021dynamics, andrade2022inferring}. A frequently used stochastic process in this context is Brownian Motion (BM), which introduces continuous random fluctuations into the transmission dynamics. BM does however have drawbacks in inference and forecasting problems. The sample paths of BM are non-stationary and thus unstable when no data is available for conditioning. We therefore prefer a stationary process which admits an equilibrium distribution as $t\rightarrow \infty$. This ensures that $\beta_t$ remains stable in the absence of data, unlike Brownian Motion, whose unbounded variance can lead to unrealistic drifts in long-term forecasts.

One other commonly used stochastic process in epidemiology modeling is the Ornstein-Uhlenbeck (OU) process which exhibits mean-reverting behavior. Theoretical studies on the existence and uniqueness of global solutions, as well as the existence of ergodic stationary distributions for SDE systems with $\beta_t$ modeled by an OU process, have drawn considerable attention. Following Wang et al.'s work on the stochastic SIS model \cite{wang2018stochastic}, several studies have explored similar theoretical aspects of other stochastic compartmental models, resulting in a growing number of publications including \cite{laaribi2023generalized, allen2017primer, zhang2023dynamics, zhou2021ergodic, cao2019dynamical}, to name a few.

While it effectively captures random noise and is stationary, the OU process does not guarantee the non-negativity of the transmission rate $\beta_t$. To address this limitation, the Black-Karasinski (BK) process, a mean-reverting stochastic process that applies a log transformation to OU process, ensuring non-negative values, has been proposed \cite{black1991bond}. Inspired by Allen's work on environmental variability \cite{allen2016environmental}, Han and Jiang \cite{han2023complete} studied the stationary distribution and local stability of a stochastic SEIR-type model using the BK process. Subsequent theoretical advancements have been made for SEIR \cite{shi2023dynamics}, SEIS \cite{zhou2024stationary}, and smoking epidemic models \cite{han2025global}. While these theoretical results suggest that OU and BK processes are potentially powerful drivers for transmission dynamics, little work has been done to validate these processes in a stochastic compartmental model with real data and studying the challenges of parameter identifiability.

Outside the field of epidemiology modeling, pMCMC has been employed to infer OU process parameters. For instance, Uyeda and Harmon \cite{Uyeda2014} used Bayesian inference to fit OU models in evolutionary biology, addressing trait evolution under stabilizing selection. In systems biology, Golightly and Wilkinson \cite{Golightly2008} utilized pMCMC for parameter estimation in multivariate diffusion processes that include the dynamics of the OU process. Van der Meulen et al. \cite{Meulen2017} focused on estimation for diffusion processes using advanced MCMC techniques related to pMCMC. These studies demonstrate the effectiveness of pMCMC in handling parameter inference for OU processes. However, the application of pMCMC method to epidemiological models that incorporate the OU process or BK process remains limited and underexplored.

This paper aims to bridge the gap between theoretical development and practical application by applying the parameter fitting approach pMCMC for the specific combination of the SIHR model with the BK process.
We outline the structure of the paper as follows. 
Section 2 details the mathematical model with notations and definitions. 
Section 3 describes the parameter estimation approach.
Section 4 presents experiments to investigate the inference challenges and applies our methodology to Arizona influenza hospitalization data.
Section 5 concludes with a summary, the main challenges encountered, and possible directions for future research.

\section{Model description}

In this section, we explicitly formulate the mathematical models and introduce notation and definitions used throughout the study.

% \textbf{The deterministic SIHR model:} 
\subsection{The deterministic SIHR model}
In this study, we focus on the SIHR model, which simplifies the complex compartmental frameworks previously used in our COVID-19 studies \cite{Ratnavale2022, Chen2021, Mallela2022, Miller2023}. The SIHR model is defined by the following system of differential equations:
\begin{equation} \label{eqn:determ_SIHR_model} 
  \left\{ \:\:
  \begin{split}
  \dfrac{dS_t}{dt} \:\: &= \: - \beta_t \frac{S_t I_t}{N}, \\
  \dfrac{dI_t}{dt} \:\: &= \: \beta_t \frac{S_t I_t}{N} - \alpha I_t, \\
  \dfrac{dH_t}{dt} \:\: &= \: \alpha \gamma I_t - \eta H_t, \\
  \dfrac{dR_t}{dt} \:\: &= \: \alpha (1 - \gamma) I_t + \eta H_t, \\
  \end{split}
  \right.
\end{equation}
where $S$, $I$, $H$, and $R$ are the susceptible, infected, hospitalized, and recovered populations respectively. The total population \( N \) is given by \( N = S_t + I_t + H_t + R_t \) for any $t$. The time-dependent parameter $\beta_t$ is the transmission rate at time $t$ and its dynamics are discussed in the subsequent sections. $\alpha > 0$ denotes the rate at which infected individuals leave the infected compartment $I$, $\gamma \in [0,1]$ is the proportion of infected individuals who become hospitalized, and $\eta > 0$ represents the recovery rate of hospitalized individuals.

%\begin{equation} \label{eqn:determ_SIHR_model} 
%\begin{cases}
%	\: \dfrac{dS}{dt} \:= - \beta_t \frac{S_t I_t}{N}, \\
%	\: \dfrac{dI}{dt} \:= \beta_t \frac{S_t I_t}{N} - \alpha I_t, \\
%	\: \dfrac{dH}{dt} \:= \alpha \gamma I_t - \eta H_t, \\
%	\: \dfrac{dR}{dt} \:= \alpha (1 - \gamma) I_t + \eta H_t, \\
%\end{cases}
%\end{equation}

% \textbf{Modeling transmission rates with stochastic processes:} 
% \subsection{Modeling transmission rates with stochastic processes}

The transmission rate, $\beta_t$, is often modeled as a time-dependent parameter, but with unknown dynamics.
Due to the unpredictable nature of disease transmission, stochastic modeling has been widely adopted to incorporate randomness and capture uncertainties in dynamic systems. Originally developed for financial applications such as stock market forecasting, stochastic models effectively handle market uncertainties and random fluctuations \cite{merton1973, black1973, hull1990}. The success of the stochastic modeling approach in finance has inspired its use in other fields, including epidemiology, where modeling uncertainty and predicting future events are equally critical \cite{allen2010introduction, keeling2008modeling}. By incorporating stochastic processes, epidemiological models more effectively capture inherent randomness and uncertainties in disease dynamics, such as demographic variability, environmental fluctuations, and shifts in human behavior \cite{allen2008introduction,nasell2002stochastic,he2010IF,king2015IF}. 
%By accounting for these sources of variability, stochastic models, such as Brownian Motion (BM),  Ornstein-Uhlenbeck (OU) process  and L\'evy processes, offer a more realistic representation of outbreak dynamics and enhance the reliability of forecasts.

\subsection{Modeling transmission rates via Brownian motion}

A common first choice is to treat the time varying transmission rate $\beta_t$ as a Brownian‐motion (BM) process governed by the stochastic differential equation (SDE):
\begin{equation}\label{eqn:BM}
 d\beta_t = \mu dt + \sigma dB_t,    
\end{equation}
where $\mu$ is the drift coefficient, representing the average change rate of $\beta_t$, $\sigma$ is the volatility parameter that controls the magnitude of random fluctuations around the long-term mean, and $B_t$ represents a standard Brownian motion.
In order to rigorously define Brownian motion, it is also necessary to introduce the notion of a complete probability space $(\Omega, \mathcal{F}, \{\mathcal{F}_t\}_{t \geq 0}, \mathbb{P})$ with a filtration $\{\mathcal{F}_t\}_{t \geq 0}$ satisfying the usual conditions (i.e., it is increasing and right-continuous, and $\mathcal{F}_0$ contains all $\mathbb{P}$-null sets). For more details on undefined terminologies, notations, and results of SDE, refer to {\O}ksendal \cite{Oksendal2003}.

Note that the increments of Brownian motion satisfy
\begin{equation*}
 dB_t \sim \mathcal{N}(0, dt),
\end{equation*}
indicating that $dB_t$'s are normally distributed with mean zero and variance $dt$. Therefore, the distribution of $\beta_t$ is 
\begin{equation}\label{eqn:BM_distribution}
 \beta_t\sim \mathcal{N}(\beta_0+\mu t, \, \sigma^2t).
\end{equation} 
This implies that, as $t\rightarrow\infty$, the mean and variance of $\beta_t$ approach infinity. To avoid this issue, $\mu$ is often set to be 0 in applications, so that $\beta_t$ is centered at $\beta_0$. However, because the variance of $\beta_t$ still increases linearly with time, the increasing uncertainty makes the process less predictable over time. This can pose challenges for long-term forecasting in epidemiology.

A variant of this stochastic equation considers the logarithm of $\beta_t$ as Brownian motion as follows:
\begin{equation*}
 d\ln{\beta_t} = \mu dt + \sigma dB_t,
\end{equation*}
which is often used to ensure that the transmission rate $\beta_t$ remains positive. In this case $\beta_t$ is called geometric Brownian motion and $\ln{\beta_t}$ follows a normal distribution, so $\beta_t$ has a log-normal distribution.

% \textbf{Advancing to the Ornstein-Uhlenbeck (OU) process:}
\subsection{The OU process and BK process}

%Moved to Intro
%Besides Brownian motion, another commonly used stochastic process in epidemiology modeling is the Ornstein-Uhlenbeck (OU) process, a continuous-time model that exhibits mean-reverting behavior --- that is, it tends to drift toward a long-term average value over time. 

Whereas Brownian motion has unbounded variance, the Ornstein-Uhlenbeck (OU) process admits an equilibrium distribution, making it a popular choice in epidemiological modeling for parameters expected to stabilize around a long-term average.
The standard OU process for $\beta_t$ can be defined by the
\begin{equation}
 d\beta_t = \lambda (\mu -  \beta_t) \, dt + \sigma \sqrt{2\lambda} \, dB_t,
\end{equation}
where $\mu$ is the long-term mean which represents the stable equilibrium value that the process tends to approach over time; $\lambda$ denotes the mean-reverting rate, which controls how quickly the process returns to the long-term mean; and $\sigma$ represents the volatility term that determines the magnitude of the random fluctuations around the long-term mean. Note that the distribution of $\beta_t$ at any time $t$ is $\beta_t \sim \mathcal{N}(\mu, \sigma^2)$, meaning that $\beta_t$ has a stationary distribution. The mean-reverting characteristic allows the OU process to better capture scenarios where values are expected to return to a baseline level, rather than exhibiting unbounded growth or divergence in BM as appearing in equation \eqref{eqn:BM_distribution}.

While the OU process effectively captures random fluctuations, it does not ensure that $\beta_t$ remains strictly positive. The Black-Karasinski (BK) process addresses this issue by applying a log transform to preserve non-negativity \cite{black1991bond}.
It is defined by the SDE
\begin{equation} \label{eqn:BK_process}
    d\ln{ \beta_t } = \lambda (\mu - \ln{ \beta_t }) \, dt + \sigma \sqrt{2\lambda} \, dB_t, 
\end{equation}
where $ \ln{ \beta_t }$ follows a standard OU process with a stationary distribution independent of $t$, $\ln{ \beta_t } \sim \mathcal{N}(\mu, \sigma^2)$. This SDE has an analytical solution
\begin{equation}\label{eqn:BK_solution}
 \beta_t
 = \beta_0{e^{-\lambda t}} \cdot 
   \exp\left( \mu (1 - e^{-\lambda t}) + \sigma \sqrt{1 - e^{-2\lambda t}} \cdot \epsilon_t \right),    
\end{equation}
where $\epsilon_t\sim \mathcal{N}(0, 1)$. We will incorporate this analytical solution into our numerical SDE solver.

\subsection{The stochastic SIHR model}

To model the disease dynamics, we combine the SIHR compartmental model \eqref{eqn:determ_SIHR_model} with the BK process \eqref{eqn:BK_process} to extend the SIHR model as follows:
\begin{equation} \label{eqn:SIHR_model} 
  \left\{ \:\:
  \begin{split}
  \dfrac{dS_t}{dt} \:\: &= \: - \beta_t \frac{S_t I_t}{N}, \\
  \dfrac{dI_t}{dt} \:\: &= \: \beta_t \frac{S_t I_t}{N} - \alpha I_t, \\
  \dfrac{dH_t}{dt} \:\: &= \: \alpha \gamma I_t - \eta H_t, \\
  \dfrac{dR_t}{dt} \:\: &= \: \alpha (1 - \gamma) I_t + \eta H_t, \\
  d\ln{\beta_t}      &= \: \lambda (\mu - \ln{\beta_t}) \, dt + \sigma \sqrt{2\lambda} \, dB_t
  \end{split}
  \right.
\end{equation}

% \begin{equation} \label{eqn:SIHR_model}
%	\begin{cases}
%		\frac{dS}{dt} & = - \beta_t \frac{S_t I_t}{N}, \\
%		\frac{dI}{dt} & = \beta_t \frac{S_t I_t}{N} - \alpha I_t, \\
%		\frac{dH}{dt} & = \alpha \gamma I_t - \eta H_t, \\
%		\frac{dR}{dt} & = \alpha (1 - \gamma) I_t + \eta H_t, \\
%		 d\ln{\beta_t} & = \lambda (\mu - \ln{\beta_t}) \, dt + \sigma \sqrt{2\lambda} \, dB_t,
%	\end{cases}
% \end{equation}
% which is a combination of the system \eqref{eqn:determ_SIHR_model} and the BK process \eqref{eqn:BK_process}.

% where
% 	\begin{itemize}
% 		\item $S_t$: Susceptible individuals at time $t$.
% 		\item $I_t$: Infected individuals at time $t$.
% 		\item $H_t$: Hospitalized individuals at time $t$.
% 		\item $R_t$: Recovered individuals at time $t$.
% 		\item $N = S_t + I_t + H_t + R_t$: Total population.
% 		\item $\beta_t$: Transmission rate at time $t$.
% 		\item $\alpha > 0$: Rate at which infected individuals leave the infected compartment ($I$).
% 		\item $\gamma \in [0,1]$: Proportion of infected individuals who become hospitalized.
% 		\item $\eta > 0$: Recovery rate of hospitalized individuals.
% 		\item $\theta > 0$: Rate of mean reversion,  Black–Karasinski process parameter.
% 		\item $\mu > 0$: Long-term mean of $\ln{\beta_t}$, Black–Karasinski process parameter.
% 		\item $\sigma > 0$: Volatility parameter, Black–Karasinski process parameter.
% 		\item $B_t$: Standard Brownian motion.
% 	\end{itemize}

The mean reverting-rate $\lambda > 0$ quantifies how quickly public health interventions, behavioral changes, or natural processes bring the transmission rate back to normal levels after disturbances. Also, $\lambda$ reflects the strength of regulatory mechanisms or feedback processes in the transmission dynamics. The long-term mean $\mu$ encapsulates the baseline transmission rate, considering factors like average contact rates, virus strains, etc. Therefore, $\mu$ serves as the central tendency around which $\ln{\beta_t}$ oscillates due to stochastic influences. The volatility $\sigma > 0$ represents the impact of random events or uncertainties affecting disease transmission, such as sudden changes in human behavior, environmental factors, and reporting errors. By introducing stochasticity into the model, $\sigma$ permits more realistic simulation of epidemic dynamics that account for unpredictability.

As this specific combination of an SIHR model with a BK process is unstudied, we first establish in Appendix Theorem \ref{Theorem:existence_and_uniqueness} that the stochastic system (\ref{eqn:SIHR_model}) is well-posed: a solution exists for any admissible initial conditions, and that solution is unique (there is exactly one possible trajectory once the random noise is fixed). This mathematical analysis ensures that the conclusions drawn from numerical simulations are reliable.
% $\sigma$ introduces stochasticity into the model, allowing for more realistic simulations of epidemic dynamics that account for unpredictability.

\section{Methods}

We will use Bayesian inference to combine the model (\ref{eqn:SIHR_model}) with data for estimating the state $X_t$ and the parameters of the BK process and some other parameters, see Figure \ref{fig:structure} for our parameter estimation framework.

\begin{figure}[!ht]
  \centering
  \includegraphics[]{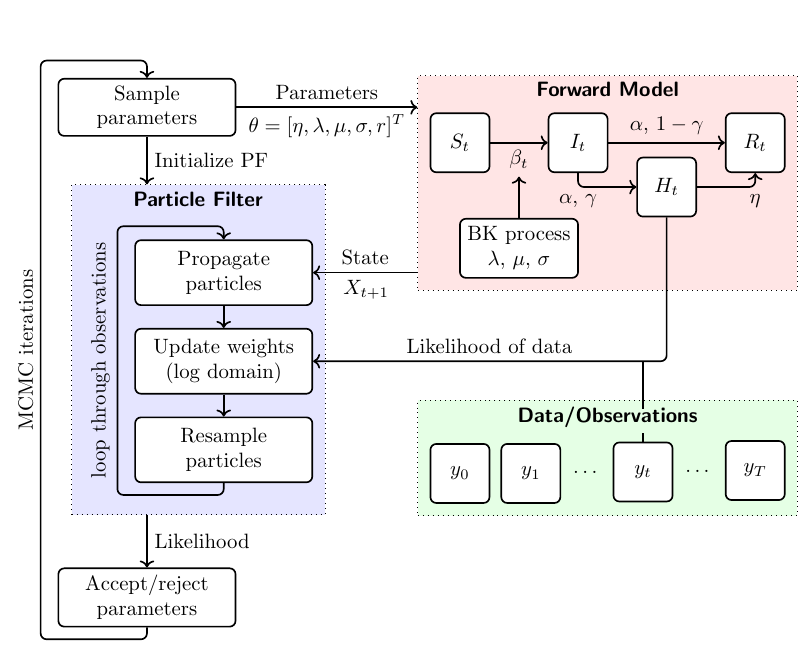}
  \caption{\label{fig:structure} Schematic diagram of the parameter estimation framework. The particle MCMC algorithm has an outer loop of MCMC iterations where each iteration samples candidate values for the parameters and then runs a particle filter to assess the likelihood of the data given the parameters. The particle filter is an inner loop which processes the data sequentially in time. At each step of the inner loop, particles (weighted samples in the state space) are propagated according to the forward model (\ref{eqn:SIHR_model}) and then updated based on the likelihood of the data.}
  \end{figure}

%\subsection{Model description}

\subsection{Inferring parameters by particle MCMC}

For notational convenience, we discretize the compartments as $S_{0:T}$, $I_{0:T}$, $H_{0:T}$, $R_{0:T}$, and the path of the transmission rate as $\beta_{0:T}$, and seek to infer its parameters $\theta = [\eta, \lambda, \mu, \sigma, r]^{\T}$. Here, $\theta$ includes the SIHR model parameters in system \eqref{eqn:determ_SIHR_model}, the parameters of the BK process governing $\beta_t$ in equation \eqref{eqn:BK_process}, and the dispersion parameter in the negative binomial distribution \eqref{eqn:NB_rp}. In our context, we are given the number of influenza hospitalizations, denoted as $y_{0:T}$, reported by public health officials or obtained from synthetic data. 

A state-space model connects a latent state process and the observed data through a probabilistic framework that incorporates both state evolution and observation densities. Based on the  dynamics of the system \eqref{eqn:SIHR_model}, we define the  state-space model as
\begin{align*}
% &y_{1:T} \sim g(Y_{1:T} \mid X_{1:T}, \theta),\\
% &X_{1:T} \sim f(\beta_{1:T}; \theta). 
 X_{0} &\sim \pi_0( \,\cdot\,; \theta ), \\
 X_{t} &\sim \psi( \,\cdot\, |\, X_{t-1}; \theta), \\ % \forall t \in \mathbb{N},
 Y_{t} &\sim h( \,\cdot\, |\, X_{t}; \theta)
\end{align*}
for $t \in \{0,1,\dots,T\}$. Here, the state variables $X_{0:T}$ contain all the compartments and the transmission rate such that $X_{0:T} = [S_{0:T}, I_{0:T}, H_{0:T}, R_{0:T}, \ln \beta_{0:T}]^{\T}$, and the observations $y_{0:T}$ are assumed to be known. The forward model $\psi$ updates the first four variables by applying a forward Euler discretization to the SIHR model in equation \eqref{eqn:SIHR_model}, and use \eqref{eqn:BK_solution} to update the last variable. In addition, the forward model $\psi$ in our study updates the last state variable $\beta_t$ using the analytical solution of the BK process defined in equation \eqref{eqn:BK_solution}.

To define the observation density $h( \,\cdot\, | X_t, \theta)$, we assume that the number of observed  hospitalizations $y_t$ at time $t$  follows a negative binomial distribution with mean equal to the number of hospitalizations $H_t$:
\begin{equation}\label{eqn:NB_rp}
 Y_t \sim \mathcal{NB}(p_t, r),
\end{equation}
where $p_t$ is computed by setting the mean of the negative binomial distribution $\frac{r(1-p_t)}{p_t}$ to be $H_t$, so $p_t=\frac{r}{r+ H_t}$. Because the system in equation \eqref{eqn:SIHR_model} has a unique global solution almost surely as proved in Appendix Theorem~\ref{Theorem:existence_and_uniqueness}, we can find the numerical solution of $H_t$ by integrating the system (\ref{eqn:SIHR_model}). The initial state $X_0$ is drawn from a prior distribution $\pi_0$. We describe the time evolution of $X_{t}$ as a Markov Process, such that the distribution of $X_{t}$ depends only on the previous state $X_{t-1}$. 

% We denote this dependence as $X_{1:T} = F(\beta_{1:T}; \theta)$ to emphasize that the solution depends on the transmission rate $\beta_t$ and can be computed for any value thereof via numerical integration based on the forward model $f$ defined in \eqref{eqn:determ_SIHR_model},\ref{eqn:BK_process}. 

We aim to estimate the posterior density $\pi(X_{0:T}, \theta \,|\, y_{0:T})$ of the latent state $X_{0:T}$ and the model parameters $\theta$, given the observed data $y_{0:T}$. To facilitate efficient sampling, we factor the posterior density as the product of the conditional posterior of the latent state given data $y_{0:T}$ and parameters $\theta$, and the marginal posterior of the model parameters $\theta$:
\begin{equation}\label{eqn:posterior_factorization}
 \pi( X_{0:T}, \theta \,|\, y_{0:T} ) = \pi(X_{0:T} \,|\, y_{0:T}, \theta) \cdot \pi(\theta \,|\, y_{0:T}).
\end{equation}
Building upon the need to efficiently sample from the joint posterior distribution $\pi(X_{0:T}, \theta  \,|\, y_{0:T})$ in \eqref{eqn:posterior_factorization}, we employ the particle Markov Chain Monte Carlo (pMCMC) algorithm in \cite{andrieu2010particle}, which integrates the PF and MCMC algorithms. Specifically, the PF algorithm approximates the first factor $\pi(X_{0:T} \,|\, y_{0:T}, \theta)$ by efficiently sampling the latent states conditioned on the observed data and fixed model parameters, while MCMC targets the second factor $\pi(\theta \,|\, y_{0:T})$ by sampling from the marginal posterior of the model parameters while integrating over the latent states. In the following, we provide detailed descriptions of the PF and pMCMC algorithms we implemented.

%Specifically, we use particle filters to estimate the intractable marginal likelihood $\pi(Y_{1:T}\mid \theta )$ by integrating over the latent variables $\beta_{1:T}$. This estimation enables us to sample the model parameters $\pi(\theta  \mid Y_{1:T})$ using MCMC methods. By iteratively sampling $\theta $ and updating the particle filter estimates, pMCMC provides a coherent framework for joint inference of both the model parameters and latent states.
%

\subsection{Particle Filter}

The Particle Filter is a sequential Monte Carlo method used for estimating the latent state of a dynamical system $X_{0:T}$ given partial and noisy observations $y_{0:T}$. It was introduced in \cite{Kitagawa1996} and has subsequently become a popular technique to perform asymptotically exact inference on state space models with highly non-linear dynamics and arbitrary noise distributions. To facilitate  inference, the PF approximates the posterior distribution of the states $\pi(X_{t} \,|\, y_{0:t},\theta)$ at time t with a set of particle realizations \( \{X_t^i\}_{i=1}^{N_p} \), where $N_p$ represents the number of particles. As $N_p \rightarrow \infty$ the discrete distribution over the particles converges to $\pi(X_{t} \,|\, y_{0:t},\theta)$. A comprehensive overview of PF can be found in the survey  by Doucet et al.\ \cite{Doucet2002}. 

The PF algorithm builds the approximation of $\pi(X_{t} \,|\, y_{0:t},\theta)$ recursively by updating the particle distribution as new data $y_t$ becomes available via two steps, updating and resampling. The update step solves the dynamical system from $t-1$ to $t$ for each particle $X_{t-1}^i$, representing the initial estimate of $X_{t}^i$ before taking into account the observation $y_t$. To incorporate the observation $y_t$,  the particles are reweighted using importance sampling, where the weights are determined by the observation density $h(y_t \,|\, X_{t}^i, \theta)$. The particles are then resampled according to these weights to approximate the filtering distribution $\pi(X_t\,|\,y_{0:t},\theta)$. 
This process is repeated at each time step to recursively approximate the full joint distribution  $\pi(X_{0:T}|y_{0:T},\theta)$. Another useful byproduct of the PF algorithm is the full data likelihood estimate $\pi(y_{0:T}|\theta)$, which is obtained by recursively approximating each predictive likelihood $\pi(y_t|y_{0:t-1},\theta) = \int h(y_t|X_t,\theta)\pi(X_t|y_{0:t-1}) dX_t$  via the sample mean of the particle weights at time $t$, and then multiplying these incremental estimates together over $t=0,\cdots,T$. More detail on the theoretical foundation of this approach is given in \cite{andrieu2010particle}. 

We choose to implement the PF described in \cite{Gentner:2018} as it represents weights in the log domain, enabling more accurate weight computations, and can avoid particle degeneracy especially when the involved distributions include exponentials or products of functions. One key component of the algorithm is the calculation of the LogSumExp function using the Jacobian logarithm algorithm, detailed further by Algorithm~\ref{alg:jacob} in the Appendix. 

\begin{algorithm}[!ht]
\caption{Particle Filter in Log Domain}\label{alg:logPF}
\begin{algorithmic}[1]
    \State \textbf{Input:} Number of particles $N_p$, number of time steps $T$, observations $y_{0:T}$, prior distribution $\pi_0$, observation density $h$, forward model $\psi$, and parameter vector $\theta$. 
    
    \State \textbf{Initialize:} 
        \State \quad Draw initial states $X_0^i \sim \pi_0(\cdot \,;\, \theta)$, $i=1,\ldots, N_p$
   
    \For{$t = 1$ to $T$}
        \For{$i = 1$ to $N_p$}
            \State $X_t^i \sim \psi(\cdot \,|\, X_{t-1}^i, \theta)$ \Comment{Forward model \eqref{eqn:SIHR_model}}
            \State $\hat{w}_t^i = \ln h(y_t \,|\, X_t^i,\theta)$ 
            \Comment{computing log likelihood}
            
            \EndFor 
        \State  $\hat{w}_t = LogSumExp\left(\{\hat{w}_t^{i}\}_{i=1}^{N_p})\right)$  \Comment{Algorithm \ref{alg:jacob}, keep only last element of output}
            \For{$i = 1: N_p$}

            \State $\hat{w}_t^i{^*} =  \hat{w}_t^i - \hat{w}_t$ \Comment{Weight normalization}

            \EndFor 

        \State Draw indices $\{\ell_i\}_{i=1}^{N_p}$ with weights $\{ \hat{w}_t^{i^*}\}_{i=1}^{N_p}$ and set  $X_t^i = X_t^{\ell_i}$   \Comment{Algorithm \ref{alg:systematic_log}}

    \EndFor
    \State \textbf{Output:} Estimated particle states $\{X_{0:T}^i\}_{i=1}^{N_p}$,  cumulative log-likelihood estimate $\ln \hat{\pi}(y_{0:T} \,|\, \theta)=\sum\limits_{t=0}^{T}(\hat{w}_{t}  - \ln N_p$).
\end{algorithmic}
\end{algorithm}

% After computing the weights \( \hat{\theta }_t^i = \log(\theta_t^i) \), we apply systematic resampling, described in \cite{doucet2001sequential,gordon1993novel} to maintain particle diversity and prevent weight degeneracy. 

\subsection{Particle Markov Chain Monte Carlo (pMCMC)}

While the Particle Filter is a powerful tool for estimating hidden states, it generally assumes that model parameters $\theta$ are known, an unrealistic assumption in practice. To jointly infer both the model parameters and the latent states, we employ the pMCMC algorithm introduced in \cite{andrieu2010particle}. To propose new parameter values in the pMCMC algorithm, we employ a random walk proposal distribution. Specifically, at each iteration $m$, we generate a candidate parameter vector $\theta^{\ast}$ by sampling from a proposal distribution $q(\theta \,|\, \theta^{(m-1)})$ centered at the current parameter estimate $\theta^{(m-1)}$. In our implementation  we take $q(\theta|\theta^{(m-1)})$ to be a multivariate Gaussian distribution $(\ref{eqn:adaptive_proposal_distribution})$
\begin{equation}\label{eqn:adaptive_proposal_distribution}
 \theta^{\ast} \sim \mathcal{N}(\theta^{(m-1)}, \Sigma^{(m-1)}),    
\end{equation}
where $\Sigma^{(m-1)}$ is the proposal covariance matrix. We employ a variant of the AM algorithm discussed in \cite{AndrieuSurvey}, where the covariance matrix $\Sigma^{(m-1)}$ is constructed to balance the trade-off between exploration and acceptance rate in the parameter space. The elements of $\Sigma^{(m-1)}$ are calibrated using adaptive methods to ensure efficient mixing of the Markov chain after a burn-in period of $M_b$ iterations in which $\Sigma^{(m-1)}$ is fixed to a diagonal matrix. This approach helps address spurious correlations in the parameters in early iterations by incorporating adaptive covariance estimates only after the Markov chain has stabilized. 
The pMCMC framework detailed in Algorithm~\ref{alg:pMCMC} leverages the Particle Filter to estimate the intractable marginal likelihood $\pi(y_{0:T} \,|\, \theta)$ pointwise by integrating over the latent states $X_{0:T}$. This estimation enables us to sample the model parameters from the posterior distribution $\pi(\theta \,|\, y_{0:T})$ using MCMC methods. During each iteration of the pMCMC algorithm, the Particle Filter provides an unbiased estimate of the likelihood for the proposed parameters $\theta^*$, enabling the MCMC sampler to effectively explore the parameter space while maintaining the correct posterior distribution. Consequently, the pMCMC algorithm offers a robust and efficient framework for inference in our state-space SIHR model, providing insights into the pathwise evolution of the BK process and the stochastic dynamics of disease transmission. 

\begin{algorithm}[!ht]
\caption{pMCMC Algorithm}\label{alg:pMCMC}
\begin{algorithmic}[1]
    \State \textbf{Initialize:} 
    \State \quad Draw initial parameters from prior $\theta^{(0)} \sim \pi_{\theta}(\cdot)$ 
    \State \quad Run Algorithm~\ref{alg:logPF} with $\theta^{(0)}$ to estimate initial log-likelihood
     \vspace{-1mm}
     \begin{equation*}
     \mathcal{\hat{L}}(\theta^{(0)}) = \log \hat{\pi}(y_{0:T} \,|\, \theta^{(0)}) + \log \pi_{\theta}(\cdot)
     \end{equation*}
    
    \vspace{-2mm}
    \For{$m = 1$ to $M$}
        \State Draw new parameters $\theta^{\ast} \sim q(\theta \,|\, \theta^{(m-1)})$ \Comment{See proposal distribution \eqref{eqn:adaptive_proposal_distribution}}
        \State Compute log-likelihood estimate $\log \hat{\pi}(y_{0:T} \,|\, \theta^*)$ \Comment{Algorithm~\ref{alg:logPF}}
        \State Compute acceptance ratio:
        \vspace{-1mm}
        \begin{equation*}
        A = \mathcal{\hat{L}}(\theta^{\ast}) - \mathcal{\hat{L}}(\theta^{(m-1)}) + \log q(\theta^{(m-1)} | \theta^{\ast}) - \log q(\theta^{\ast} \,|\, \theta^{(m-1)})
        \end{equation*}

        \vspace{-2mm}
        \If{$A \geq 0$}
            \State Accept $\theta^{\ast}$: $\theta^{(m)} \gets \theta^{\ast}$
        \Else
            \State Draw $u \sim \mathcal{U}(0,1)$
            \If{$\log(u) < A$}
                \State Accept $\theta^{\ast}$: $\theta^{(m)} \gets \theta^{\ast}$
            \Else
                \State Reject $\theta^{\ast}$: $\theta^{(m)} \gets \theta^{(m-1)}$
            \EndIf

        \EndIf
    \EndFor
    \State \textbf{Output:} Estimated posterior samples $\{\theta^{(m)})\}_{m=1}^M$. 
\end{algorithmic}
\end{algorithm}

\section{Numerical experiments and results}

Having detailed the pMCMC algorithm and its integration with our stochastic SIHR model, we now proceed to evaluate its performance through a series of numerical experiments designed to validate its efficacy and limitations. These experiments include tests on synthetic data to assess the estimation accuracy of the key parameters and demonstrate the practical applicability of our approach using real-world influenza hospitalization data. A crucial aspect of these evaluations is addressing parameter identifiability issues, which occur when multiple parameter sets yield similar model outputs, complicating the task of uniquely estimating parameters from data. 

\begin{Experiment} \label{ex:pMCMC_full_model}
Estimation accuracy on synthetic data.
\end{Experiment}
This experiment aims to validate the estimation accuracy of the pMCMC algorithm on the stochastic SIHR model using synthetic data. We simulated 250 days of hospitalization data $H_{0:250}$ using the parameter set listed in Table~\ref{table:example_1}. To improve the tractability of the parameter inference problem, we fixed two key parameters: the hospitalization rate $\gamma$ and the infection recovery rate $\alpha$.

\begin{table}[!ht]
    \begin{center}
    \caption{Parameter set for the synthetic data generation and pMCMC inference}
    \label{table:example_1}
    \renewcommand{\arraystretch}{1.1}
    \setlength{\tabcolsep}{8pt}
    \begin{tabular}{llcc}
        \hline
        Parameter & Description & Value & Prior \\
        \hline
        \multicolumn{4}{l}{\textit{Model parameters}:} \\
        \( N \) & total population & 1,000,000  & fixed \\
        $\beta_0 $ & initial transmission rate & 0.4 & $\mathcal{U}(0,1)$  \\
        \( I_0\) &  initial infectious individuals & 100 & $\mathcal{U}(0,1,000)$  \\
        \( \gamma \) & hospitalization rate & $\frac{1}{1000}$ & fixed \\
        \( \eta \) & hospitalization recovery rate & $\frac{1}{10}$ & $\mathcal{U}(0,1)$ \\
        $\alpha$ & infection recovery rate & $\frac{1}{7}$ & fixed \\
        \( \sigma \) & volatility & 0.4 & $\mathcal{B}(1.5,10)$ \\
        \( \lambda \) & mean reverting rate & $\frac{1}{35}$ & $\mathcal{B}(3,10)$ \\
        $\mu$ & long term mean & -1.3 & $\mathcal{N}(-0.8,0.4)$\\
        $r$ & dispersion parameter & 100 & $\mathcal{U}(\frac{1}{1000},\frac{1}{5})$\\
        \hline
        \multicolumn{4}{l}{\textit{pMCMC parameters}:} \\
        $M$ & number of MCMC iterations & 100,000& \\
        \( N_p \) & number of particles & 1,000 &  \\
        $M_b$ & burn-in iterations & 1,000& \\
        \hline
    \end{tabular}
    \end{center}
\end{table}

% We assumed that noise in new case detection and reporting on the \(i\)-th day is captured by a negative binomial distribution \(\mathrm{NB}(r, q_i)\) centered on \(H(t_i)\), the expected number of new cases detected over the \(i\)-th day given by the SDE model. We assume that day-to-day fluctuations in the random variable are independent and characterized by a negative binomial distribution \(\mathrm{NB}(r, q_i)\), which has two parameters, \(r > 0\) and \(q_i \in (0,1)\).
% Note that $\mathbb{E}[\mathrm{NB}(r, q_i)] = r \frac{1 - q_i}{q_i}.$
% We assume that this distribution has the same dispersion parameter \(r\) across all case reports.
% With these assumptions, we arrive at the following likelihood function:
% \[
% \mathcal{L}(\theta; {H_i\}_{i=0}^d) = \prod_{i=0}^a \mathcal{L}_i(\theta_F; \delta C_i)
% \]
% where
% \[
% \mathcal{L}_i(\theta; H_i) = \mathrm{nbinom}(H_i ; r, q_i) = \binom{H_i + r - 1}{H_i - 1} q_i^r (1 - q_i)^{H_i}
% \]
% and $q_i = \frac{r}{r + H_i}$.

After generating the synthetic data, we applied the pMCMC algorithm to infer the BK process parameters and estimate the latent states which includes sample path of $\beta_t$. The Bayesian prior distributions were intentionally chosen to be weakly informative to test the robustness of the pMCMC algorithm. In Table~\ref{table:example_1}, $\mathcal{N}$, $\mathcal{U}$, and $\mathcal{B}$ refer to the normal distribution, the uniform distribution, and the beta distribution, respectively. Post pMCMC sampling, we ran the PF one more time conditioned on the posterior mean $\bar{\theta}$ of the static parameters to obtain an estimate of the posterior over $X_{0:T}$, $\pi( X_{0:T} \,|\, y_{0:T},\bar{\theta})$.

The results of a single experiment, as summarized in Figure~\ref{fig:pMCMC_results_full_exp_good}, indicate that the pMCMC algorithm successfully inferred the path of $\beta_t$. Specifically, Figure~\ref{fig:pMCMC_results_full_exp_good}(a) shows the true data remain largely contained within the shaded credible interval, suggesting that the Bayesian inference effectively captures the underlying uncertainty and variability in the data. Subplot (b)  illustrates the convergence of the MCMC chain. Subplot (c) is the posterior distribution of a key parameter $\eta$ of the SIHR model inferred by pMCMC, with the posterior mean closely aligned with the true value for the synthetic data. Subplot (d) shows the posterior distribution for the dispersion parameter $1/r$; the posterior mean deviates substantially from the true value, suggesting that the inference for this parameter is less reliable—potentially due to its weak identifiability. It is widely recognized in the statistical literature that dispersion parameters can be challenging to estimate, especially when the dispersion is close to the mean \cite{Anscombe1949, Clark1989}. Subplot (e) demonstrates that the 50\% quantile of the estimated $\beta_t$ aligns well with the true path, suggesting a good fit. Additionally, subplots (f) and (g) show that the mean ($\mu$) and volatility ($\sigma$) parameters were accurately inferred. However, subplot (h) reveals difficulties in estimating the mean reverting rate $\lambda$, indicating potential challenges in accurately capturing this parameter, which will be discussed in the following experiments.

\begin{figure}[!ht]
\centering
\includegraphics[trim={5.5cm 1.5cm 4.5cm 1.5cm},clip,width=1\textwidth]{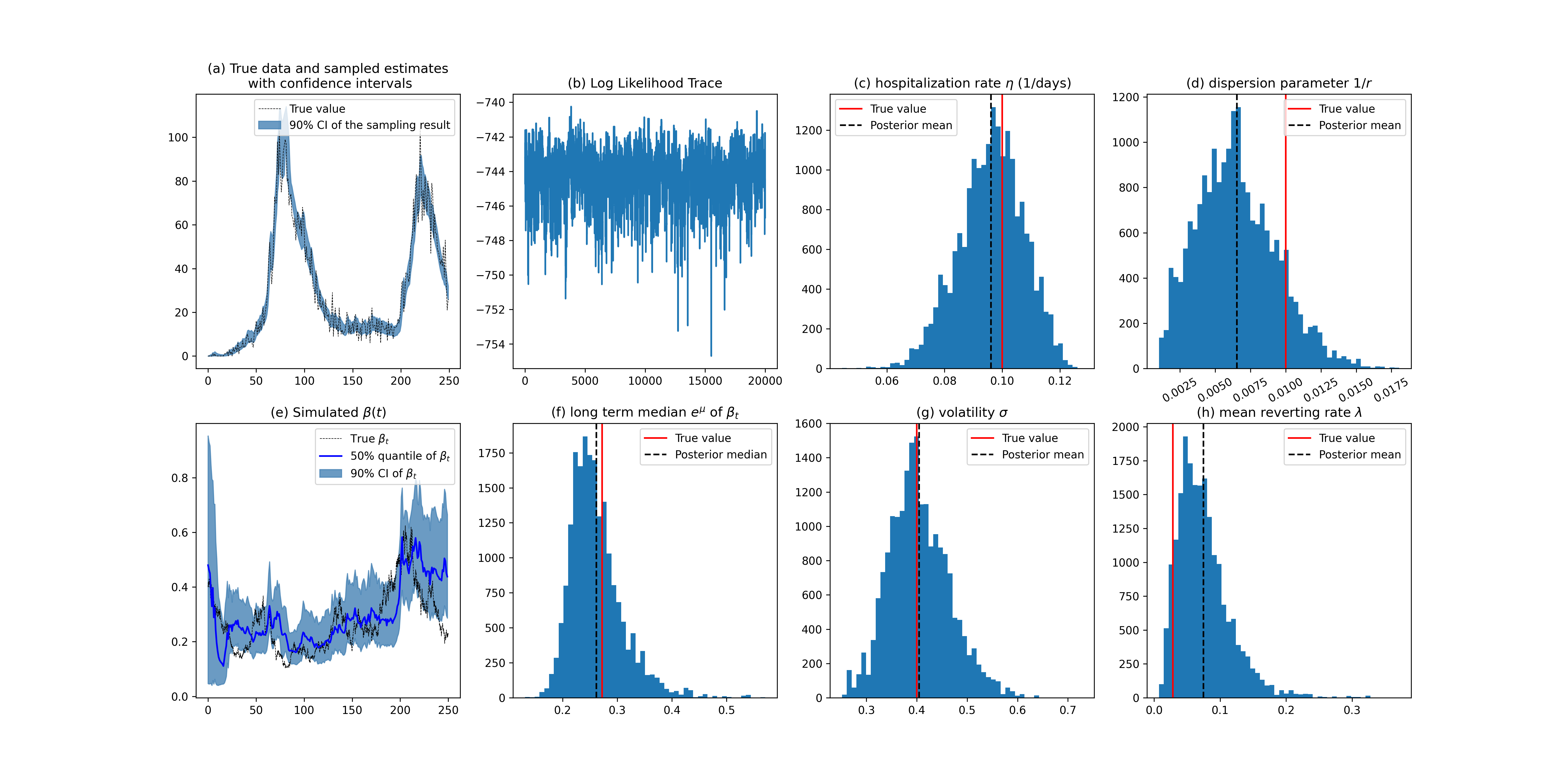}
\caption{The posterior distributions in subplots (a) and (e) are obtained conditional on the posterior mean of the static parameters.  (a) True data and sampled estimates with 90\% credible interval  (CI) obtained from the distribution of the $H$ compartment in the particle filter (PF).
(b) Trace plot of the log likelihood from the MCMC sampling process; (c) and (d) Histogram for the posterior distributions of the SIHR model parameters with the red line indicating the true value and the black dashed line showing the posterior mean; (e) True trajectory of $\beta_t$ with the 50\% quantile and 90\% CI obtained from the PF samples, conditional on the posterior mean of the static parameters; (f), (g), and (h) Histograms for the posterior distributions of the BK process parameters.}
\label{fig:pMCMC_results_full_exp_good}
\end{figure}

While the pMCMC algorithm generally performs well, we observed challenges in accurately inferring $\beta_t$ during periods with low case numbers, see Figure~\ref{fig:pMCMC_results_full_exp_bad}. Specifically, when case counts diminished, the likelihood function became insensitive to the parameters, leading to incorrect estimations. To address this, we focused our analysis on periods where the case number exceeded five. Figure~\ref{fig:pMCMC_results_full_exp_bad} highlights this interval and shows that the parameter estimates are close to the true values within the interval. Subplots (f) and (g) display the sample mean and standard deviation of the $\beta_t$ path within the high case number interval, demonstrating improved estimation accuracy. Similar to the previous results, estimating the mean reverting rate $\lambda$ remained challenging (subplot h).

\begin{figure}[!ht]
\centering
\includegraphics[trim={5.5cm 1.5cm 4.5cm 1.5cm},clip,width=1\textwidth]{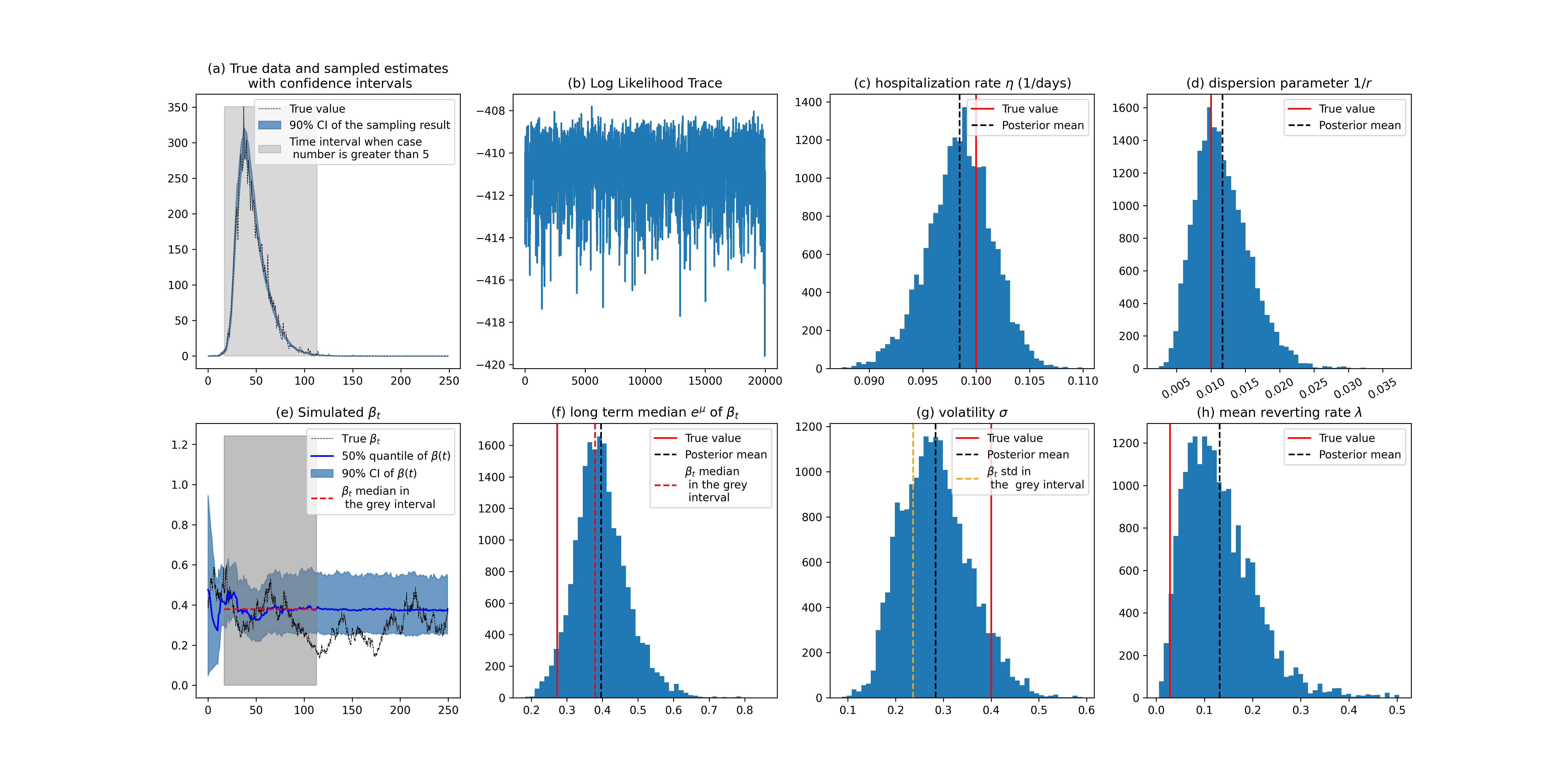}
\caption{The posterior distributions in subplots (a) and (e) are obtained conditional on the posterior mean of the static parameters.  (a) True data and sampled estimates with 90\% credible interval  (CI) obtained from the distribution of the $H$ compartment in the particle filter (PF). The shaded grey region highlights the time interval when the case number exceeds 5. (b) Log likelihood trace plot from MCMC, indicating the convergence of the MCMC sampler. (c) and (d) Histogram showing the posterior distribution of the SIHR model parameters with the red line indicating the true value and the black dashed line showing the posterior mean. (e) Simulated $\beta_t$ over time, with the true path, 50\% quantile, and 90\% CI from the PF samples. The red dashed line shows the mean of $\beta_t$ within the grey-shaded interval. (f), (g), and (h) Histograms display the posterior distributions of the BK process parameters.}
\label{fig:pMCMC_results_full_exp_bad}
\end{figure}

In summary, the experiment highlights both the strengths and limitations of the pMCMC algorithm in inferring the path of $\beta_t$ and estimating the associated parameters. The algorithm performs well when case numbers are sufficiently large, accurately capturing the dynamics of $\beta_t$ and reliably estimating parameters such as $\mu$ and $\sigma$. However, during periods of low case counts, the likelihood function's insensitivity leads to unreliable parameter estimates. Additionally, inferring the mean reverting rate $\lambda$ remains challenging throughout, suggesting potential identifiability issues that require further investigation. In the next experiment, we will explore these challenges in more detail.

\begin{Experiment}\label{ex:lambda_correct}
Effect of decorrelation time on estimation accuracy for $\beta_t$.
\end{Experiment}

A critical parameter within the pMCMC inference framework is the mean reverting rate $\lambda$, which influences the temporal dependencies in the system by dictating how quickly the process reverts to its mean. Accurate estimation of $\lambda$ within the BK process has been challenging in our experiments, which results in weak identifiability of the trajectory of $\beta_t$ in various examples.

In real epidemiological data, the true mean reverting rate $\lambda$ is often unknown. To investigate the identifiability issues associated with this parameter and to determine the optimal mean reverting rate settings in the absence of prior knowledge, we fixed $\lambda$ to a set of values $\big\{1, \frac{1}{7}, \frac{1}{7 \times 2}, \ldots, \frac{1}{7 \times 14} \big\}$, corresponding to decorrelation times ranging from 1 day to 98 days. For each decorrelation time, we generated 50 different BK processes, each representing a sample path of $\beta_t$ over 250 days. We applied the PF directly to infer the path of $\beta_t$ while fixing the static parameters at their true values $\theta$. This avoids the extra uncertainty that pMCMC can introduce through simultaneous estimation of those static parameters.

\begin{figure}[!h]
\centering
\includegraphics[width=0.64\textwidth]{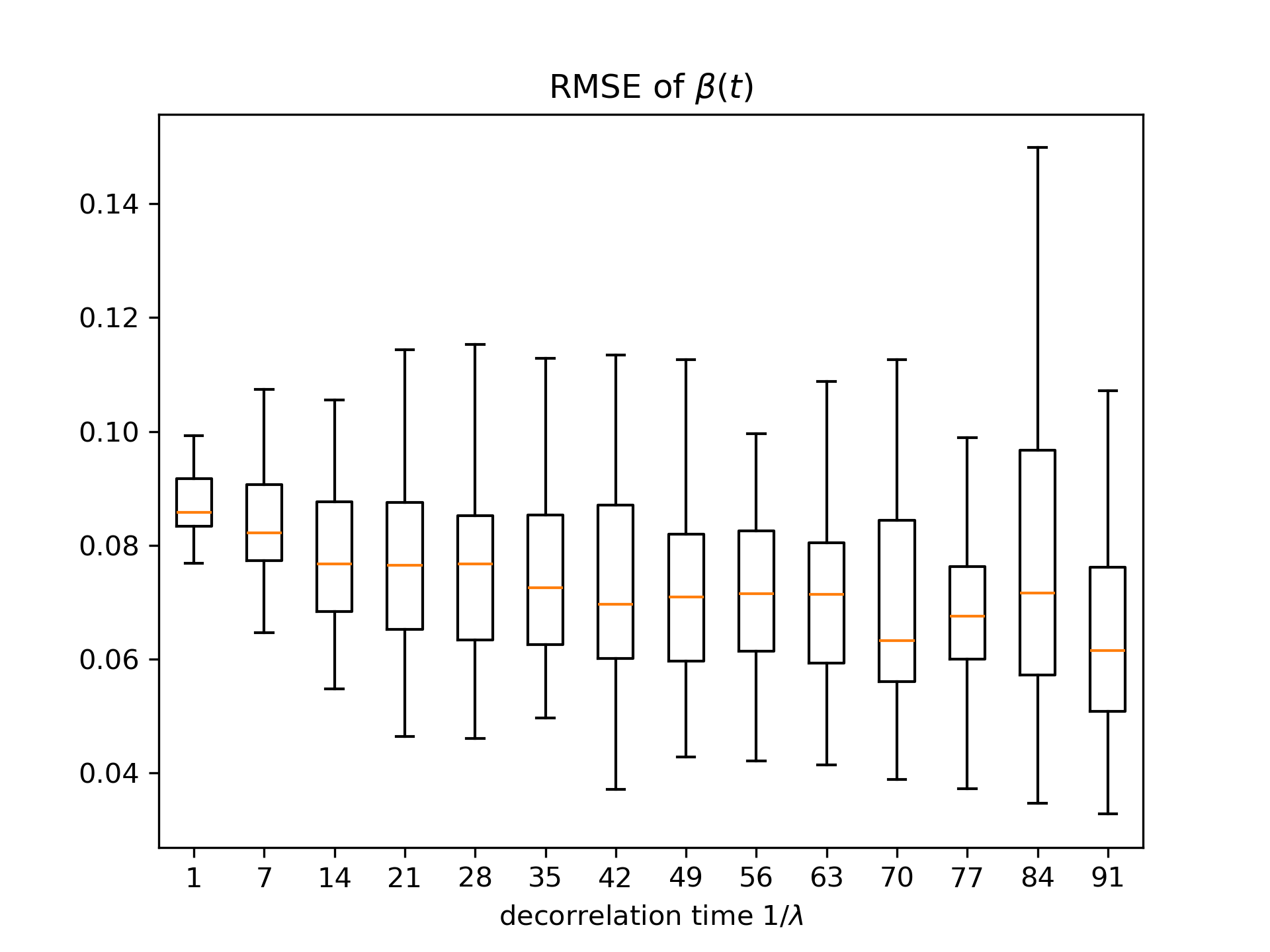}
\caption{The box plots show the root mean square error of $\beta_t$ estimates across different decorrelation times. For each decorrelation time, 50 different $\beta_t$ paths were simulated. The PF was then used to infer the trajectory of $\beta_t$,  and the RMSE was summarized for each simulations.}
\label{fig:RMSE_lambda_all}
\end{figure}

Our analysis, as depicted in Figure~\ref{fig:RMSE_lambda_all}, reveals that if $1/\lambda$ is small, then the resulting SDE has a large drift term $ \lambda (\mu - \ln{ \beta_t })$ in \eqref{eqn:BK_process}, which makes it harder for the algorithm to identify the path. The mean RMSE values remain consistent with large $1/\lambda$, indicating stable average performance of the PF regardless of $1/\lambda$. Varying $\lambda$ does not significantly impact the mean accuracy of the $\beta_t$ path estimation when decorrelation time $1/\lambda$ is large. Despite this, we observe an increase in the variance of RMSE as the decorrelation time increases. A longer decorrelation time allows $\beta_t$ to exhibit a greater variability, leading to a higher variability in the estimation errors. Consequently, while the average accuracy remains unaffected, the reliability of the estimates decreases with increasing decorrelation time, and further experiments are designed in the following to study the idetifiability issue of the algorithm.

\begin{Experiment}\label{ex:lambda_wrong}
Sensitivity of $\beta_t$  estimation to misspecification of $\lambda$.
\end{Experiment}

In this experiment, we evaluate how the deviations from the true mean-reversion rate $\lambda$ affect the performance of the PF in estimating the path of $\beta_t$. By intentionally misconfiguring $\lambda$ and analyzing the resulting estimation errors, we aim to understand the robustness of the PF against parameter misspecification and to identify an optimal $\lambda$ that maintains inference accuracy of $X_t$ and the remaining model parameters. The findings will inform best practices for parameter selection in an SIHR model with the BK process.

For each $1/\lambda \in \{ 1, 7, 14, \dots, 98 \}$ days, we generated 10 distinct datasets with BK processes, each representing $\beta_t$ over a 250-day span, as illustrated in Example~\ref{ex:pMCMC_full_model}. For each dataset, we applied the PF using all $\lambda$ values in $\big\{1, \frac{1}{7}, \frac{1}{7 \times 2}, \ldots, \frac{1}{7 \times 14} \big\}$ to infer the $\beta_t$ path, while keeping other static parameters at their true values to eliminate additional uncertainty.

\begin{figure}[!h]
\centering
\includegraphics[trim={3.5cm 0cm 3.5cm 0cm},clip,width=1\textwidth]{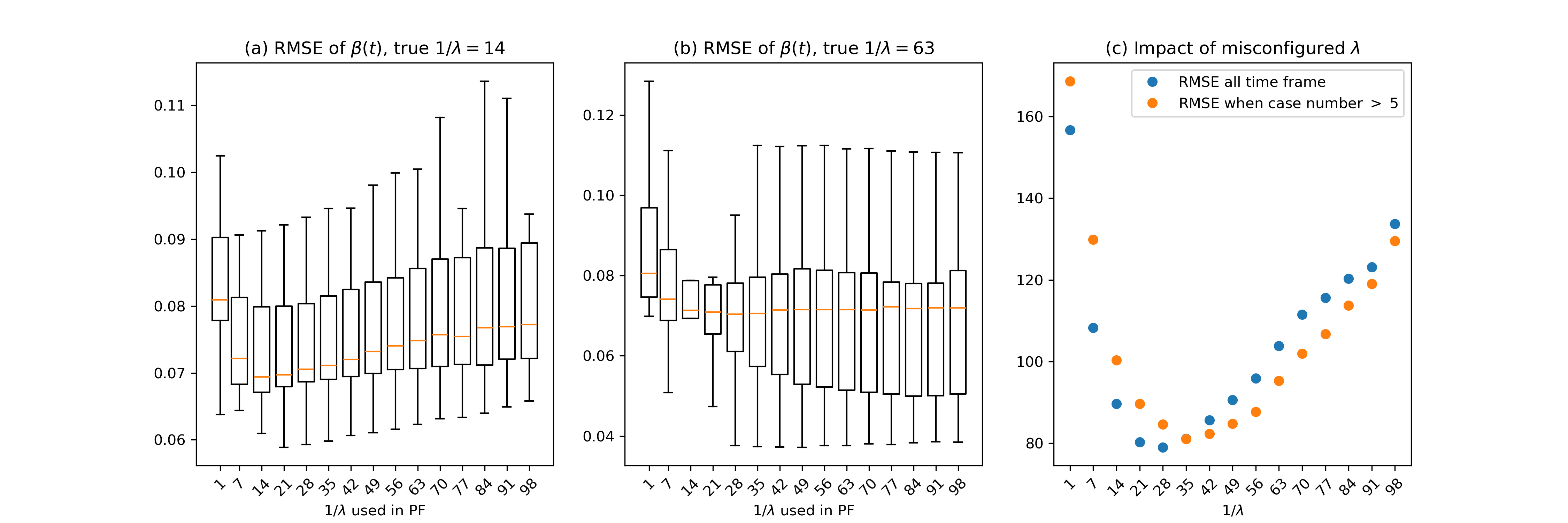}
\caption{(a) RMSE of $\beta_t$ estimates when the true decorrelation time is 14 days ($1/\lambda = 14$). (b) RMSE of $\beta_t$ estimates when the true decorrelation time is 63 days ($1/\lambda = 63$). (c) Aggregate ranking of RMSE accuracy scores across all misconfigured $\lambda$ values. The rankings are based on RMSE calculated during periods with case numbers exceeding five, indicating that a decorrelation time of approximately 35 days yields the best overall performance.}
\label{fig:wrong_lambda_box_rank}
\end{figure}

Figure~\ref{fig:wrong_lambda_box_rank}(a) shows that when the true decorrelation time is $1/\lambda = 14$ days, the mean RMSE is lowest when the PF uses the correct $\lambda$. However, the large variance suggests that the PF's performance is not highly sensitive to the misspecification of $\lambda$. This insensitivity is more pronounced in Figure~\ref{fig:wrong_lambda_box_rank}(b), with a true decorrelation time of $1/\lambda = 63$ days. Here, the mean RMSE remains relatively consistent for $1/\lambda \geq 14$ days, indicating that when the true mean reverting rate is small, it allows the path of $\beta_t$ to diffuse more freely, which in turn makes it easier for the particle filter to infer the true trajectory. Results of additional experiments showing similar effects can be found in Appendix Figure~\ref{fig:wrong_lambda_RMSE_box_all} and Figure~\ref{fig:wrong_lambda_RMSE_box_k}.

To quantify the overall performance across all configurations, we ranked the RMSE values for each $\beta_t$ path and aggregated these ranks for each decorrelation time setting. The aggregated ranking, presented in Figure~\ref{fig:wrong_lambda_box_rank}(c), show that a decorrelation time of approximately 35 days yielded the best mean estimation accuracy. This suggests that, in the absence of accurate knowledge of the true $\lambda$, the setting of $1/\lambda = 35$ days offers a reasonable balance between estimation error and variability. These findings highlight the PF's robustness to certain degrees of misspecification in $\lambda$, especially when the true decorrelation time is large. %They underscore the importance of selecting an appropriate default value for $1/\lambda$ when it cannot be precisely estimated, with 35 days emerging as a practical choice for minimizing estimation errors in $\beta_t$.%

\begin{Experiment}
Application to real data—modeling influenza hospitalizations in Arizona. 
\end{Experiment}

Having solidified our confidence in model behavior using synthetic data, we then turned to validation of our model using real world observational data.   Specifically, we used the latest daily hospital admissions data at the state level in the U.S. obtained from healthdata.gov, an official source provided by the CDC.
The data was obtained from reporting of daily hospital admissions which was mandatory from February 2022 until it was suspended on May 1, 2024. Since then, reporting has been voluntary through the CDC’s National Healthcare Safety Network and is has been provided at weekly resolution. In our analysis, we focus on 20 weeks of daily hospitalization case counts for Arizona during two periods: from October 1, 2022 to February 13, 2023, and from October 1, 2023 to February 13, 2024. These time spans were selected because they contain a reasonably large number of hospitalization cases, suggested by Experiment \ref{ex:pMCMC_full_model}.

We employed the pMCMC algorithm in conjunction with the SDE model to infer the static model parameters, including the hospitalization recovery rate $\eta$, the rate of progression from infection $\alpha$, and the BK process parameters $\mu$ and $\sigma$. In this experiment, we changed the likelihood function from negative binomial distribution to Poisson distribution because the dispersion parameter inferred from the dataset was large ($>1000$). A large $r$ implies that the variance of the negative binomial distribution approaches the mean, making Poisson distribution a suitable approximation. In addition, we fixed the mean reverting rate $1/\lambda=35$ suggested by Experiments~\ref{ex:lambda_correct} and \ref{ex:lambda_wrong}.

\begin{figure}[!h]
\centering
\includegraphics[trim={4.6cm 0cm 4.6cm 0cm},clip,width=1\textwidth]{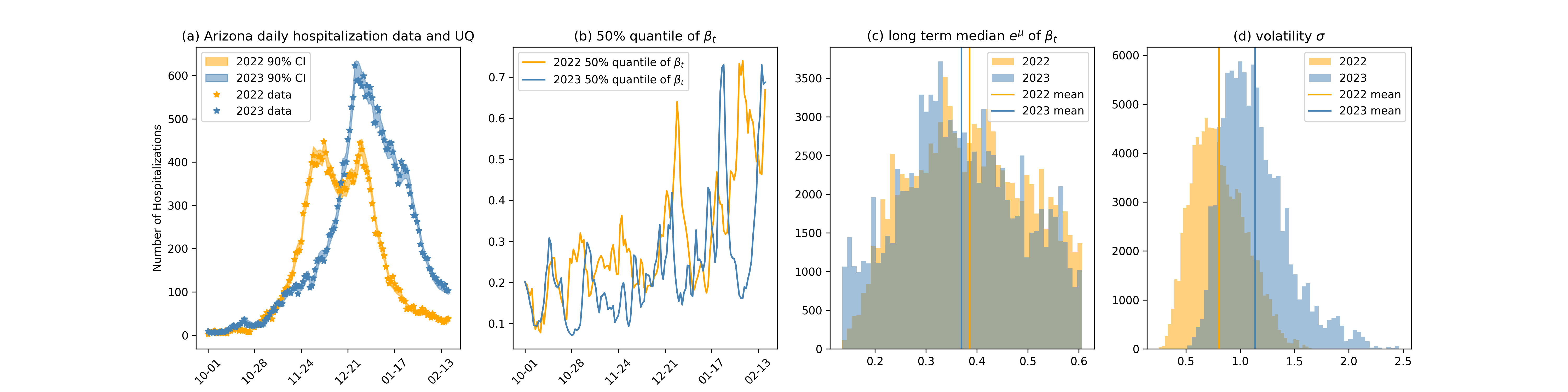}
\caption{(a) Comparison of Arizona daily hospitalization data and associated uncertainty quantification (UQ) for 2022 and 2023. (b) Median estimates of temporal patterns in the transmission rate $\beta_t$. (c) and (d) histograms of the posterior distributions for long-term mean $\mu$ and volatility $\sigma$  respectively, with vertical lines denoting mean values for each year.}
\label{fig:Arizona_flu}
\end{figure}

Figure~\ref{fig:Arizona_flu} summarizes our results. As indicated in Figure~\ref{fig:Arizona_flu}(a), the model's estimates align closely with the observed hospitalization data, indicating the efficiency of the inference framework. Figure~\ref{fig:Arizona_flu}(b) shows that the inferred transmission rate $\beta_t$ exhibits similar patterns in both 2022 and 2023, with an increasing trend around December. This is consistent with the CDC's report in \cite{CDC2024}, which indicates that influenza activity in both years began to rise in early November and peaked in late December -- a typical trend for influenza seasons. Figure~\ref{fig:Arizona_flu}(c) shows that the posterior distribution of the long-term mean $\mu$ of $\beta_t$ for 2022 and 2023 are comparable, as documented in \cite{CDC2024} that the percentage of specimens testing positive for influenza and the cumulative rates of influenza-associated hospitalizations were similar in both years. Furthermore, consistent public health policies and community behaviors during these periods likely contributed to the similar transmission dynamics. However, the volatility $\sigma$ of the transmission rate, as displayed in Figure~\ref{fig:Arizona_flu}(d), is notably higher in 2023. According to the CDC's report, this variability may be attributed to co-circulating influenza strains and a shift in predominant virus types. Specifically, the increased activity of influenza B viruses in February 2024 likely influenced transmission dynamics, contributing to the increased volatility.

In addition, the  mean hospitalization recovery rate $\eta$ is estimated to be 0.28 with 95\% credible interval [0.22,0.32] based on data from 2022–2023, and 0.22 with 95\% credible interval [0.18,0.24] based on data from 2023–2024, corresponding to an average hospitalization length of 3–5.5 days. This finding is consistent with the survey results reported in \cite{Pivette2020,Lina2020}, which indicate that the length of hospitalization for influenza ranges from 3.4 to 11.5 days depending on the age group. The estimated mean progression rate from infection $\alpha$ is 0.13 with with 95\% credible interval [0.09,0.2] based on data from 2022-2023, and 0.1 with 95\% credible interval [0.06,0.15] based on data from 2023-2024, corresponding to an average progression length from 5 to 16 after symptom onset, which is supported by a study from the Global Influenza Hospital Surveillance Network specified that patients aged over 5 years had to exhibit at least one systemic symptom (such as fever, malaise, headache, or myalgia) and one respiratory symptom (like cough, sore throat, or shortness of breath) and must have been hospitalized within 7 days of symptom onset to be included in the study \cite{Lina2020}.

\section{Conclusion}

In this study, we developed a stochastic SIHR model with the BK process which models the transmission rate $\beta_t$.  The selection of the SIHR model was driven by the availability of US influenza hospitalization data. 

We initially employed the pMCMC algorithm to determine whether the $\beta_t$ path, the BK process parameters, the dispersion parameter, and the SIHR model parameters could be accurately inferred. The results demonstrated that accurate inference was achievable when case numbers remained above a certain threshold. Under these conditions, the data exhibited sufficient signal strength, and the likelihood function was highly sensitive to changes in model parameters, allowing the algorithm to distinguish among different parameter values and converge to the true ones.

In contrast, some experiments failed to recover the correct sample path when case numbers were near zero. In these scenarios, the likelihood function became insensitive to parameter changes, resulting in inaccurate estimates. By restricting the analysis to periods with at least five reported cases, the parameter estimates aligned closely with the true values. These findings highlight the importance of maintaining a minimum case threshold to ensure identifiability and reliable inference.

In these experiments, we also found that the algorithm exhibits reduced sensitivity to the mean-reverting rate $\lambda$. So we further investigate the identifiability issue on $\lambda$ and found that the PF within the pMCMC framework is less responsive to the variation in $\lambda$, particularly when $\lambda$ is small. Despite the misspecification in $\lambda$, the PF could still accurately estimate the $\beta_t$ path. Given the complex real-world factors influencing the mean-reverting rate, such as human behavior and environmental conditions, it is challenging to determine an optimal $\lambda$. Based on our experiments, we selected a decorrelation time of $1/\lambda = 35$ days which empirically minimized the adverse effects of $\lambda$ misspecification. 

Finally, we applied the SIHR model and pMCMC algorithm to U.S. influenza hospitalization data from Arizona for the two-year flu seasons spanning from October to February in 2022-2023 and 2023--2024. The inferred $\beta_t$ trajectory closely mirrored the transmission rate fluctuations observed in Arizona. This alignment can be attributed to consistent levels of influenza activity and public health measures across the two-year flu seasons. In particular, the increased volatility in 2023 can be attributed to the co-circulation of multiple influenza strains. These results underscore the model's ability to capture real-world epidemiological dynamics.

In conclusion, our study demonstrates the efficacy of integrating the stochastic SIHR model using a BK process with pMCMC for robust parameter inference and state estimation in epidemiological modeling. The BK process provides certain theoretical stability guarantees that give it an advantage over non-stationary processes such as Brownian motion in the low data regime; we can also more effectively infer process parameters such as the mean and variance which are useful analyzing the long term dynamics of $\beta_t$.  Future research could explore the incorporation of alternative compartmental models or stochastic processes. The SIHR model is a reasonable choice for modeling influenza over a single season, where we typically assume that recovered individuals do not become susceptible again (since reinfections within one season are relatively rare). However, this assumption can cause the susceptible compartment to rapidly diminish, artificially inflating the estimated transmission rate at the end of the season. One key reason for this phenomenon is the lack of sufficient stochasticity in the state dynamics: without it, the model compensates for observed fluctuations by pushing $\beta_t$ dramatically to match observed cases exactly. We plan to implement the PF described in  \cite{Calvetti2021}, which incorporates process noise directly into the compartment transitions and absorbs some of the observational variance. In addition, the BK process may not be ideal for modeling the transmission rate, especially when the available data do not provide sufficient information to accurately estimate the long-term mean. Under those circumstances, $\beta_t$  could be biased toward an unrealistic value.

\section*{Data availability}

All models, data, and code required to reproduce our results and figures are available at \url{https://github.com/NAU-CCL/PMCMC/tree/main}.

\section*{Declaration of competing interest}

The authors declare that they have no known competing financial interests or personal relationships that could have appeared to influence the work reported in this paper.

\section*{Acknowledgements}

This research was supported by NIH under Grant No. 2U54MD012388 and R01GM111510. Additionally, Avery Drennan gratefully acknowledges the support of the Jennison and Parks scholarship. Ye Chen would also like to express her sincere gratitude to Perko Faculty Research Award. Jaechoul Lee's research was supported by Northern Arizona University's Technology \& Research Initiative Fund (TRIF). Computational analyses were run on Northern Arizona University's Monsoon computing cluster, funded by Arizona’s Technology and Research Initiative Fund

\appendix
\section{Appendix}

\subsection{Existence and uniqueness of a global solution}
In this section we rigorously prove that the model satisfies the existence and uniqueness criteria, and specify the parameter fitting and state estimation algorithms.
In the study of stochastic dynamical systems, the existence of a unique global solution ensures that the conclusions drawn from numerical simulations and analyses are reliable. We aim to prove that for any initial condition within the domain, the system admits a unique solution that remains positive for all time with probability one. We start with a proof of Lemma 1, which generalizes the It\^{o}'s Lemma to study the dynamics of $V(X_t, t)$ in a system governed by the SDE $d X_t = f(X_t) d t + g(X_t) d B_t$. The proof employs standard results on stochastic integration and quadratic variation as developed in \cite{Oksendal2003}, while tailoring the analysis to the specific structure of our problem under consideration.

Let $X_t$ be a $d$-dimensional It\^{o} process governed by the stochastic differential equation
\begin{equation*}
d X_t = f(X_t) d t + g(X_t) d B_t, \qquad X_t \in \mathbb{R}^d,
\end{equation*}
where $f: \mathbb{R}^d \rightarrow \mathbb{R}^d$, $g: \mathbb{R}^d \rightarrow \mathbb{R}^{d}$, and $B_t$ is a standard $m$-dimensional Brownian motion. Suppose $V(X_t, t) \in C^{2,1}(\mathbb{R}^d \times \mathbb{R}_{+}, \mathbb{R})$ is a scalar-valued function that is continuously differentiable in $t$ and twice continuously differentiable in $X_t$. Then, the dynamics of $V(X_t, t)$ are given by
\begin{equation*}
d V(X_t, t) = \mathcal{L} V(X_t, t) d t + V_X(X_t, t) g(X_t) d B_t,
\end{equation*}
with the differential operator $\mathcal{L}$ is defined as
\begin{equation*}
 \mathcal{L} V(X_t, t)
 = V_t(X_t, t) + V_X(X_t, t) f(X_t) + \frac{1}{2} g(X_t)^{\T} V_{X X}(X_t, t) g(X_t),
\end{equation*}
where $V_t = \frac{\partial V}{\partial t}$, $V_X = \left[ \frac{\partial V}{\partial x_1}, \ldots, \frac{\partial V}{\partial x_d} \right]$, and the Hessian matrix $V_{XX} = \left[ \frac{\partial^2 V}{\partial x_i \partial x_j} \right]_{d \times d}$.

To work directly with \(\beta_t\), we transform the equation \eqref{eqn:BK_process} using  It\^{o}'s lemma. Let \(u_t = \ln(\beta_t)\), so \(\beta_t = e^{u_t}\). Then
\begin{align*}
d\beta_t &= e^{u_t} \, du_t + \frac{1}{2} e^{u_t} \, (du_t)^2 \\
&= \beta_t \left[ \lambda (\mu - \ln(\beta_t)) \, dt + \sigma \sqrt{2\lambda} \, dB_t \right] + \frac{1}{2} \beta_t \cdot 2\lambda \sigma^2 \, dt \\
&= \beta_t \left[ \lambda (\mu - \ln(\beta_t) + \sigma^2) \, dt + \sigma \sqrt{2\lambda} \, dB_t \right].
\end{align*}
Thus
\begin{equation} 
d\beta_t = \beta_t \lambda (\mu - \ln(\beta_t) + \sigma^2) \, dt + \beta_t \sigma \sqrt{2\lambda} \, dB_t. \label{eq:beta}
\end{equation}
with an initial condition $\beta_0 $. Here, $\lambda > 0$ denotes the rate of mean reversion, $\mu$ is the long-term mean, $\sigma$ represents the volatility, and $B_t$ is a standard Brownian motion. Note that $ \ln{\beta_t}$ is a standard Ornstein–Uhlenbeck process with stationary distribution: $\ln{\beta_t}\sim \mathcal{N}(\mu, \sigma^2)$.

With the equation \eqref{eq:beta}, the system of SDE \eqref{eqn:SIHR_model} can be re-expressed as
\begin{equation*}
 dX_t = f(X_t) \, dt + g(X_t) \, dB_t,
\end{equation*}
where $X_t$, $f(\cdot)$, and $g(\cdot)$ are
\begin{equation*}
 X_t
 = \begin{bmatrix} 
   S_t \\ I_t \\ H_t \\ R_t \\  \beta_t \\
   \end{bmatrix}, \quad\:
 f(X_t)
 = \begin{bmatrix}
-\beta_t \frac{S_t I_t}{N} \\
\beta_t \frac{S_t I_t}{N} - \alpha I_t \\
\alpha \gamma I_t - \eta H_t \\
\alpha (1 - \gamma) I_t + \eta H_t \\
\beta_t \lambda (\mu - \ln(\beta_t) + \sigma^2)
   \end{bmatrix}, \quad\:
 g(X_t)
 = \begin{bmatrix}
   0 \\ 0 \\ 0 \\ 0 \\ \beta_t \sigma \sqrt{2\lambda} \\
  \end{bmatrix}.
\end{equation*} 

\begin{Lemma} \label{lemma:continuous_ae}
    The Black–Karasinski process $\{X_t, t \geq 0\}$ defined in \eqref{eqn:BK_process} is continuous on $(0,\infty)$ almost everywhere.
\end{Lemma} 

\begin{proof}
For a rigorous proof of the continuity of the Black–Karasinski process, refer to \cite{Oksendal2003}.
\end{proof}

We will show that the stochastic SIHR model satisfies the Lipschitz condition. The proof involves proving that both the drift term $f(X)$ and the diffusion term $g(X)$ are Lipschitz continuous.

\begin{Lemma}\label{lem:local Lipschitz}
The system in \eqref{eqn:SIHR_model} is locally Lipschitz on $\mathbb{R}_+^4 \times [a,\infty)$ with $a>0$  almost everywhere. 
\end{Lemma}

\begin{proof}
   Let \( X = [S, I, H, R, \beta]^{\top} \) and \( Y = [S', I', H', R', \beta']^{\top} \).  As \( g(X_t) \) is linear in \( \beta_t \), we have
\begin{equation*}
    \| g(X) - g(Y) \| = |\beta_t \sigma \sqrt{2\lambda} - \beta'_t \sigma \sqrt{2\lambda}| = \sigma \sqrt{2\lambda} |\beta - \beta'| \leq \sigma \sqrt{2\lambda} \| X - Y \|,
\end{equation*}
so \( g \) is globally Lipschitz.

For \( f(X_t) \), consider a compact set \( K \subset \mathbb{R}_+^4 \times [a, \infty) \), where \( |S|, |I|, |H|, |R|, \beta \leq M \) and \( \beta \geq a > 0 \). Denote \(f_1, f_2, \cdots, f_5\) as the five rows of $f(X)$.  Let $X, Y\in K$.

Then, for the $S$ compartment, we have
\begin{equation*}
|f_1(X) - f_1(Y)| = \left| -\beta \frac{S I}{N} + \beta' \frac{S' I'}{N} \right| \leq \frac{M^2}{N} (|\beta - \beta'| + |S - S'| + |I - I'|). 
\end{equation*}
For the $I$ compartment, we have
\begin{equation*}
|f_2(X) - f_2(Y)| \leq \frac{M^2}{N} (|\beta - \beta'| + |S - S'| + |I - I'|) + \alpha |I - I'|.
\end{equation*}
For the $H$ compartment, we have
\begin{equation*}
|f_3(X) - f_3(Y)| = |\alpha \gamma (I - I') - \eta (H - H')| \leq \alpha \gamma |I - I'| + \eta |H - H'|.
\end{equation*}
For the $R$ compartment, we have
\begin{equation*}
|f_4(X) - f_4(Y)| = |\alpha (1 - \gamma) (I - I') + \eta (H - H')| \leq \alpha (1 - \gamma) |I - I'| + \eta |H - H'|.
\end{equation*}
For $\beta$, define \( h(\beta) = \beta \lambda (\mu - \ln(\beta) + \sigma^2) \). As $h(\beta)$ is continuous a.e., we have \(  h'(\beta)= \mu +  \sigma^2 -1 - \ln(\beta)  \) is bounded a.e., say by \( K_h \), so 
\begin{equation*}
|f_5(X) - f_5(Y)| \leq \lambda K_h |\beta - \beta'|.
\end{equation*}

With the bounds from each component and the \( L^1 \)-norm, we have
\begin{equation*}
\| f(X) - f(Y) \|_1 \leq C ( |S - S'| + |I - I'| + |H - H'| + |R - R'| + |\beta - \beta'| ) = C \| X - Y \|_1,
\end{equation*}
where \( C \) depends on \( M, N, \alpha, \gamma, \eta, \lambda, K_h \). Thus, \( f \) is locally Lipschitz a.e., and the system satisfies the local Lipschitz condition on \( \mathbb{R}_+^4 \times [a, \infty) \) a.e..

\end{proof}

\begin{Theorem}\label{Theorem:existence_and_uniqueness}
For any initial value $[S(0), I(0), H(0), R(0), \beta_0]^{\T} \in \mathbb{R}_+^4 \times [a,\infty)$ with $a>0$, there exists a unique solution $X_t = [S_t, I_t, H_t, R_t, \beta_t]^{\T}$ of the model \eqref{eqn:SIHR_model} on $t \geq 0$, and the solution remains in $\mathbb{R}_+^4 \times [a,\infty)$ almost surely.
\end{Theorem}

\begin{proof}
Since $S_t + I_t + H_t + R_t = N$ for any $t \geq 0$, we can reduce the system by substituting $R$ with $R = N - S - I - H$. By Lemma~\ref{lem:local Lipschitz} in Appendix, the system is locally Lipschitz on $\mathbb{R}_+^4 \times [a,\infty)$ a.s., so the system has a unique local solution $X_t$ on $t \in (0, \tau_e]$ a.e., where $\tau_e$ is an explosion time. To establish global existence, we define a sequence of stopping times
\begin{equation*}
 \tau_l = \inf \{ 0\leq t \leq \tau_e : X_t \notin F_l \}, \qquad l \geq l_0,
\end{equation*}
where $F_l = (-l, l)^4\times [\frac{1}{l},l)$, and for any given $X(0) \in \mathbb{R}_+^4 \times [a,\infty)$, there exists a sufficiently large $l_0$ such that $X(0) \in F_{l_0}$. Note that $\tau_l$ is an increasing sequence of $l$. Denote $\tau_\infty = \lim_{l \to \infty} \tau_l$. Without loss of generality, we assume that $\inf\{\emptyset\} = +\infty$. Then, $\tau_\infty \leq \tau_e$ a.s. We only need to prove $\tau_\infty = +\infty$ a.s., which will imply $\tau_e = +\infty$ and $X_t \in \mathbb{R}_+^4 \times [a,\infty)$ a.s. for $\forall t > 0$.

We prove this by contradiction. If we assume $\tau_\infty < +\infty$ a.s., then there exists $\varepsilon_0 \in (0, 1)$, $T_0 > 0$, and $N_0 > 0$ such that
\begin{equation}\label{eqn:contra_epsilon0}
 \mathbb{P}(\tau_l \leq T_0) \geq \varepsilon_0, \qquad \text{for } \forall \: l \geq N_0.
\end{equation}
%This was my original choice of Lyapunov function. However, it doesn't work as the derivative is not bounded above. It will work if beta is generated by OU process, but not BK process.
%\[    V(S, I, H, \ln{\beta}) = (S - 1 - \ln S) + (I - 1 - \ln I) + (H - 1 - \ln H) + \frac{1}{2\lambda} (\ln{\beta} - \mu)^2.\] \[ V(S, I, H, \ln{\beta}) = (S - 1 - \ln S) + (I - 1 - \ln I) + (H - 1 - \ln H) + (\beta - 1 - \ln \beta ) \]
Consider the Lyapunov function $V(S, I, H, \beta)$ defined as:
\begin{equation*}
 V( S, I, H, \beta ) = (S-1-\ln S)+(I-1-\ln I)+(H-1-\ln H)+(\beta-1-\ln \beta),
\end{equation*}
which is non-negative and $V(S, I, H, \beta) \in C^2(\mathbb{R}_{+}^4 \times[a, \infty), \mathbb{R})$. As the drift term of $\ln \beta$ is $-\lambda(\ln \beta-\mu)$, and the diffusion term is $\sigma \sqrt{2 \lambda}$, by It\^{o} calculus, the dynamics of $V(S, I, H, \beta)$ satisfy:
\begin{equation*}
 d V(S, I, H, \beta)
 = \mathcal{L} V(S, I, H,  \beta) d t + (\beta-1) \sigma \sqrt{2 \lambda} \, d B_t,
\end{equation*}
where the stochastic process's generator $\mathcal{L}$ is given by:
\begin{align*}
 \mathcal{L} V(S, I, H, \beta)
 &= \left( 1 - \frac{1}{S} \right) \left( -\beta \frac{SI}{N} \right)
  + \left( 1 - \frac{1}{I} \right) \left( \beta \frac{SI}{N} - \alpha I \right)  \\
 &\quad + \left( 1 - \frac{1}{H} \right) (\alpha \gamma I - \eta H) 
  + \left(1 - \frac{1}{\beta}\right) \beta \lambda (\mu - \ln \beta + \sigma^2) \\
&\quad + \frac{1}{2} \cdot \frac{1}{\beta^2} \cdot (\beta \sigma \sqrt{2\lambda})^2\\
 &= -\beta \frac{S I}{N} + \beta \frac{I}{N} + \beta \frac{S I}{N} - \alpha I - \beta \frac{S}{N} + \alpha + \alpha \gamma I - \eta H - \frac{\alpha \gamma I}{H} + \eta \\
&\quad + \lambda \beta (\mu - \ln \beta + \sigma^2) - \lambda (\mu - \ln \beta + \sigma^2) + \lambda \sigma^2\\
& \le \beta \left( \frac{I - S}{N} \right) - \alpha (1 - \gamma) I - \eta H - \frac{\alpha \gamma I}{H} + \alpha + \eta \\
&\quad + \lambda \beta (\mu - \ln \beta + \sigma^2) - \lambda (\mu - \ln \beta + \sigma^2) + \lambda \sigma^2\\
& \le \beta + \lambda \beta (\mu - \ln \beta + \sigma^2) + \alpha + \eta - \lambda (\mu - \ln \beta + \sigma^2) + \lambda \sigma^2.
\end{align*}
As \(\beta \to \infty\), \(-\lambda \beta \ln \beta\) dominates, and as \(\beta \to 0^+\), \(- \lambda (\mu - \ln \beta + \sigma^2)\) ensures \(\mathcal{L}V\) remains finite. Thus, \(\mathcal{L}V \leq M\) for some constant \(M\).
\begin{equation*}
 d V(S, I, H, \beta) \leq M dt + (\beta - 1)\sigma\sqrt{2\lambda} \, d B_t
\end{equation*}
for some $M\in \mathbb{R}$. Integrate this inequality from $0$ to $\tau_l \wedge T_0$, use the boundedness of $\mathcal{L}V$, and take the expectation:
\begin{equation} \label{eqn:E_V}
 \mathbb{E}[
  V( S(\tau_l \wedge T_0), I(\tau_l \wedge T_0), H(\tau_l \wedge T_0), \ln(\beta) (\tau_l \wedge T_0) ) ] 
 \leq MT_0 + V(S(0), I(0), H(0), \ln(\beta_0)).
\end{equation}
On the other hand, $\forall$ path $\omega \in \{\tau_l \leq T_0\}$, at least one of $S(\tau_l, \omega)$, $I(\tau_l, \omega)$, $H(\tau_l, \omega)$, $\ln{\beta}(\tau_l, \omega)$ is no less than $l$. By combining \eqref{eqn:contra_epsilon0} and \eqref{eqn:E_V}, we have
\begin{align*}
  M T_0 + V(S(0), I(0), H(0), \beta_0)
  &\geq \mathbb{E}\left[ V( S(\tau_l \wedge T_0), I(\tau_l \wedge T_0), H(\tau_l \wedge T_0), \ln{\beta}(\tau_l \wedge T_0) ) \right] \\
  &\geq \mathbb{E}\left[ 1_{\{\tau_l \leq T_0\}} V( S(\tau_l, \omega), I(\tau_l, \omega), H(\tau_l, \omega), \ln{\beta}(\tau_l, \omega) ) \right] \\
  &\geq \mathbb{P}(\tau_l \leq T_0) \cdot V( S(\tau_l, \omega), I(\tau_l, \omega), H(\tau_l, \omega), \ln{\beta}(\tau_l, \omega) ) \\
  &\geq \varepsilon_0 \left[(l - 1 - \ln l) \land \left(\frac{1}{l} - 1 + \ln l \right) \right].
\end{align*}
Taking $l \to \infty$, this leads to the contradiction
\begin{equation*}
 \infty = M T_0 + V( S(0), I(0), H(0), \ln{\beta_0} ) < \infty.
\end{equation*}
Therefore, we must have $\tau_l = \infty$ almost surely. This completes the proof of Theorem \ref{Theorem:existence_and_uniqueness}.
\end{proof}

\subsection{Algorithms for resampling in the log domain}

Resampling in particle filtering transforms a set of weighted samples from the posterior distribution $\pi(X_{0:t} \,|\, y_{0:t},\theta)$ to an unweighted set by duplicating high-weight samples and discarding low-weight samples. This process improves sampling efficiency and enhances the estimation of high-dimensional integrals in Bayesian filtering.

To improve numerical stability when dealing with very small weights, which is a common situation due to observation densities with highly concentrated probability mass, we perform resampling in the log domain. Performing computation  in the log domain is a standard numerical technique that prevents numerical underflow and overflow issues by working with logarithms of density functions instead of the density functions themselves. We employ a variant of the systematic resampling algorithm that was proposed by Gentner et al.~\cite{Gentner:2018} to allow all resampling computations to be performed strictly in the log domain. We selected systematic resampling due to its $\mathcal{O}(n)$ runtime and the low variance of the resulting samples.

A central challenge in log-domain resampling is computing the normalization constant for the distribution of the log weights $\hat{w}_t^i$. This is efficiently addressed using the Iterative Jacobian Logarithm algorithm (Algorithm~\ref{alg:jacob}), which computes the logarithm sums of the form $\log( \sum_{i=1}^{N} \exp( \hat{w}_t^i ))$ without directly exponentiating the log weights. This function is also known as the LogSumExp function and well studied in the machine learning community for use in neural network architectures.

\begin{algorithm}[!ht]
\caption{Iterative Jacobian Logarithm}\label{alg:jacob}
\begin{algorithmic}[1]
   \State \textbf{Input:} $\{ \hat{w}_t^i\}_{i=1}^{N_p}$

   \State \textbf{Initialize:} $C_t^1 = \hat{w}_t^1$

   \For{$i = 2:N_p$}
     \State $C_t^i = \max(\hat{w}_t^i,C_t^{i-1}) + \log(1+\exp(-|\hat{w}_t^i - C_t^{i-1}|)) $\Comment{$C_t^i = LogSumExp(\{\hat{w}_t^k\}_{k=1}^i)$} 
   \EndFor

   \State \textbf{Output:} $\{C_t^i\}_{i=1}^{N_p}$
\end{algorithmic}
\end{algorithm}

The iterative nature of Algorithm~\ref{alg:jacob} facilitates the computation of both the log normalization constant and the log cumulative distribution function (CDF) of the particle weights. The log CDF is essential for the systematic resampling algorithm in the log domain, as presented in Algorithm~\ref{alg:systematic_log}. By performing all calculations in the log domain, we effectively sidestep numerical instability issues associated with underflow and overflow.

\begin{algorithm}[!ht]
\caption{Systematic Resampling in Log Domain}\label{alg:systematic_log}
\begin{algorithmic}[1]
    \State \textbf{Input:} Normalized weights $\{\hat{w}_{t}^{i*}\}_{i=1}^{N_p}$ 
    \State \textbf{Initialize:} $k = 1$,$\{\ell_i = i\}_{i=1}^{N_p}$  \Comment{$\{\ell_i\}_{i=1}^{N_p}$ is initialized as the index set $\ell_1 = 1, \ell_2 = 2, \dots$} 
    \State Compute $\{C_t^i\}_{i=1}^{N_p} = LogSumExp(\{\hat{w}_{t}^{i*}\}_{i=1}^{N_p})$ \Comment{ Algorithm \ref{alg:jacob}}. 
    
    \State Draw $s \sim \mathcal{U}(0,\frac{1}{N_p})$
      \For{$i = 1:N_{p}$}
        \vspace{1mm}
        \State $U^{i} = \log \big( s + \frac{i}{N_p} \big)$
        \vspace{1mm}
        \While{$U^{i} > C_{t}^{k}$}
            \State $k = k + 1$
      \EndWhile\label{sysendwhile}
      \State $\ell_{i} = k$
      \EndFor
       \State \textbf{Output:} $\{\ell_i\}_{i=1}^{N_p}$

  \end{algorithmic}
\end{algorithm}

By utilizing these algorithms, we perform resampling entirely in the log domain, enhancing numerical stability and efficiency in the PF. This approach enables efficient sampling of the state space and avoids the particle degeneracy issues associated with very small weights.

\begin{figure}[!ht]
\centering
\includegraphics[width=1\textwidth]{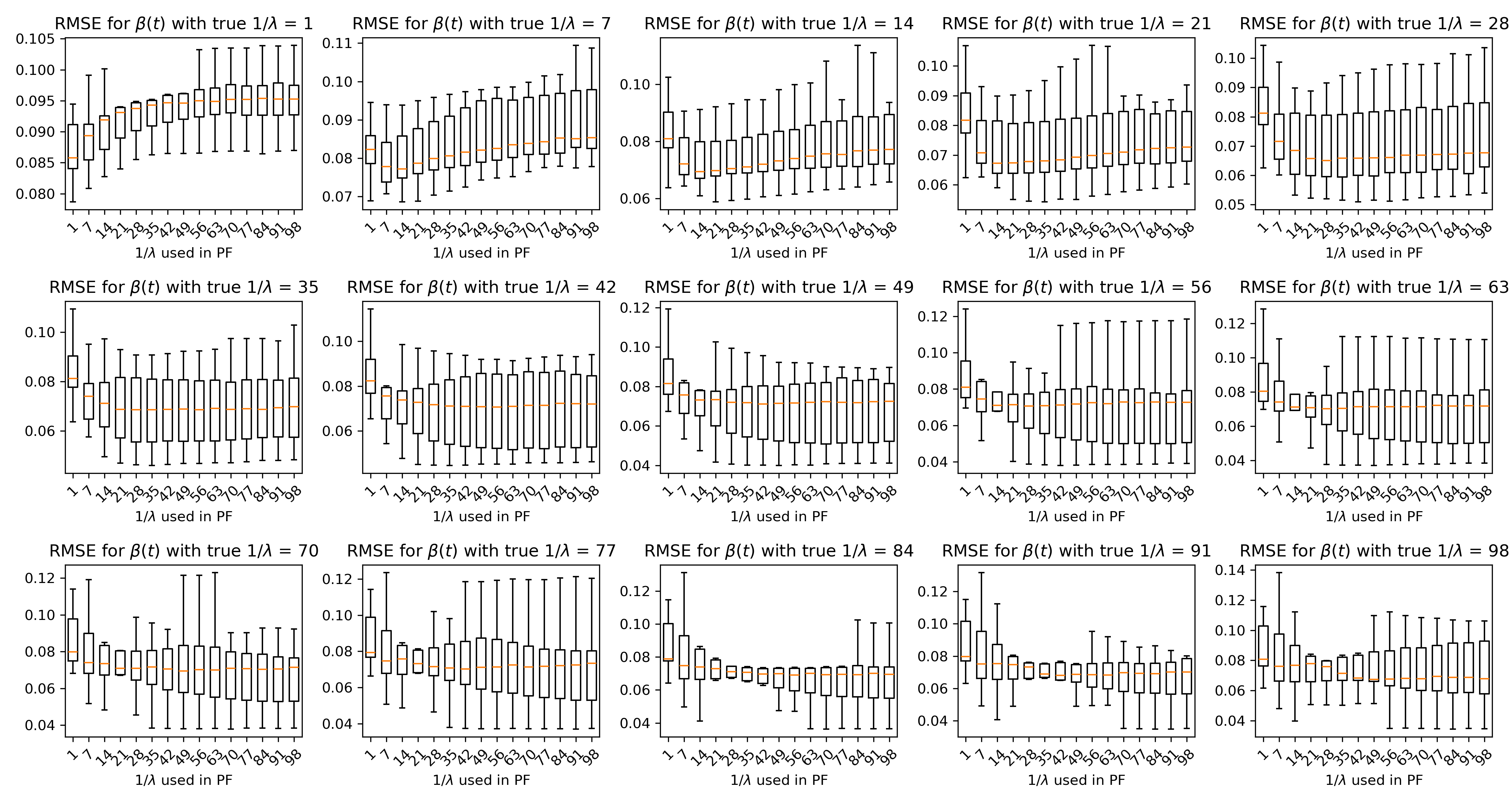}
\caption{RMSE accuracy scores with true and misconfigured $\lambda$ values. For each true $\lambda$, ten different $\beta_t$ path and data set were simulated as described in Experiment~\ref{ex:pMCMC_full_model}. Then, each $1/\lambda \in \{ 1, 7, \ldots, 98 \}$ was used in PF to infer the path of $\beta_t$ with the simulated data. The RMSE accuracy scores were calculated for the inferred $\beta_t$. The results demonstrate that when the decorrelation time $1/\lambda$ is 28 days or longer, the impact of misconfiguring $1/\lambda$ to values exceeding 28 days is minimal. This suggests that setting $1/\lambda \geq 28$ days serves as a robust default, maintaining high estimation accuracy of $\beta_t$ even when $\lambda$ is not precisely inferred.}
\label{fig:wrong_lambda_RMSE_box_all}
\end{figure}

\begin{figure}[!ht]
\centering
\includegraphics[width=1\textwidth]{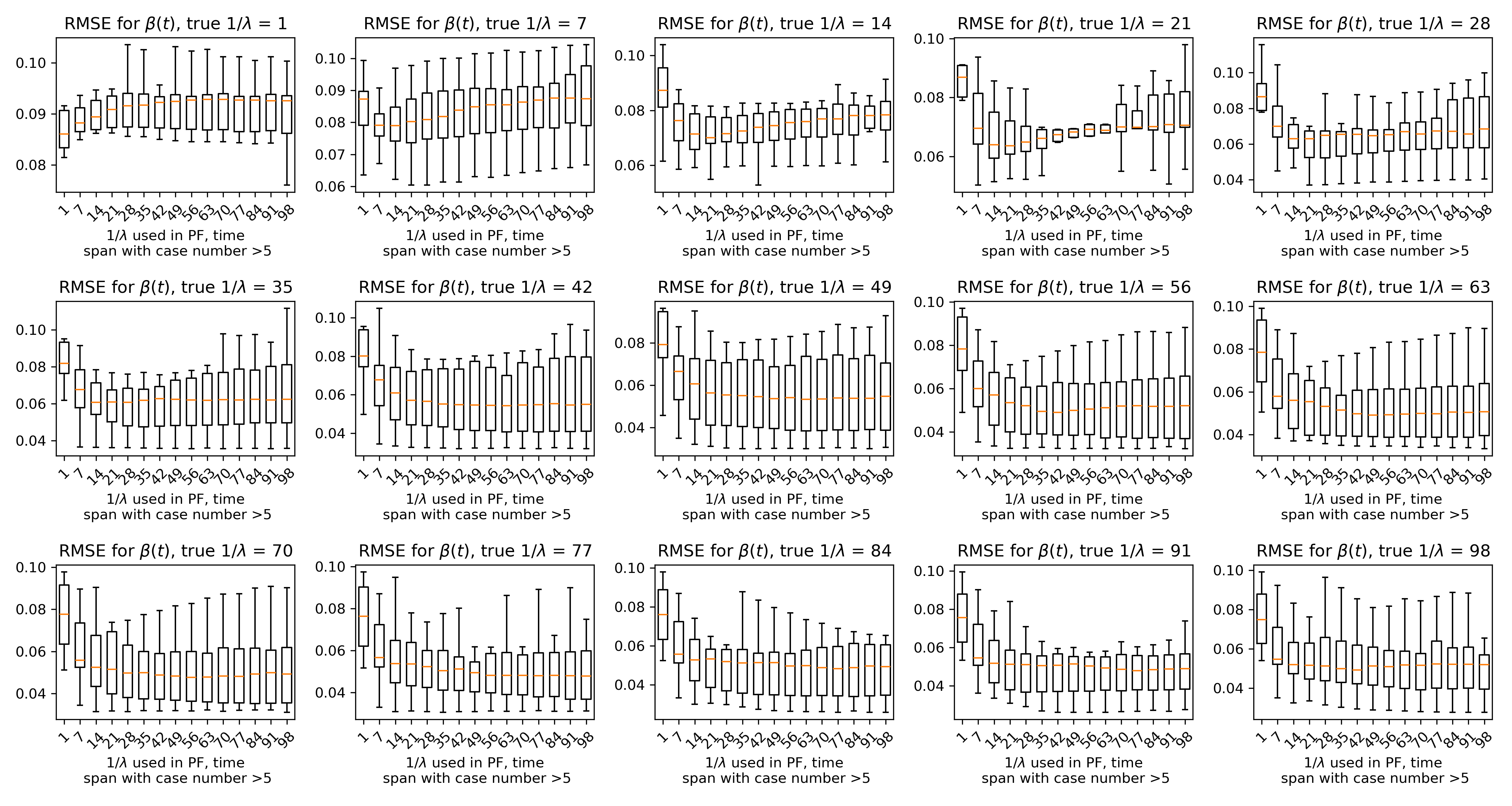}
\caption{RMSE accuracy scores with true and misconfigured $\lambda$ values for case number $\geq$ 5. For each true $\lambda$, ten different $\beta_t$ path and data set were simulated as described in Experiment~\ref{ex:pMCMC_full_model}. Then, each $1/\lambda \in \{ 1, 7, \ldots, 98 \}$ was used in PF to infer the path of $\beta_t$ with the given simulated data. The RMSE accuracy scores were calculated for the inferred $\beta_t$ over the time interval when the case number is greater than or equal to 5. The results demonstrate that when the decorrelation time $1/\lambda$ is 35 days or longer, the impact of misconfiguring $1/\lambda $ to values exceeding 35 days is minimal. This suggests that setting $1/\lambda \geq 35$ days serves as a robust default, maintaining high estimation accuracy of $\beta_t$ even when $\lambda$ is not precisely inferred.}
\label{fig:wrong_lambda_RMSE_box_k}
\end{figure}

\clearpage

\newpage

\bibliography{NIH}{}
\bibliographystyle{nihunsrt} % Use the custom nihunsrt bibliography style included with the template

\end{document}